\documentclass[journal]{IEEEtran}

\IEEEoverridecommandlockouts

\usepackage{xcolor,soul,framed} %,caption
\usepackage{bm}
\usepackage[font=small,labelfont=bf]{caption}
\colorlet{shadecolor}{yellow}
\usepackage[pdftex]{graphicx}
\graphicspath{{../pdf/}{../jpeg/}}
\DeclareGraphicsExtensions{.pdf,.jpeg,.png}
\usepackage{import}
\usepackage{amsmath}
\usepackage{amssymb}
\usepackage[linktocpage=true]{hyperref}
\usepackage{hyperref}
\usepackage{multicol,lipsum}
\usepackage{array}
\usepackage{xcolor}
\usepackage{color}
\usepackage{makecell}
\usepackage{tikz}
\usepackage{blkarray}
\usepackage{amsmath,amsfonts,amssymb}
\usepackage{mathtools}
\usepackage{cuted}
\usepackage{cite}
\usepackage{cleveref}
\usepackage{amssymb}
\usepackage{cases}
\usepackage{empheq}
\usepackage{titlesec}

\newcommand{\tb}{\textcolor{black}}
\newcommand{\mt}{\mathtt}

\newcommand{\mb}{\mathbf}

\newcommand{\N}{\mathsf{N}}

\newcommand{\x}{\bm{x}}
\newcommand{\1}{\mathtt{A}}
\newcommand{\2}{\mathtt{B}}

\newcommand{\p}{\mathsf{p}}

\newcommand{\y}{\bm{y}}

\newcommand{\V}{\mathcal{V}}

\newcommand{\X}{\mathcal{X}}

\def\mathbi#1{\textbf{\em #1}}
\newtheorem{theorem}{Theorem}
\newtheorem{lemma}{Lemma}
\newtheorem{assumption}{Assumption}
\newtheorem{remark}{Remark}
\newtheorem{cor}{Corollary}
\newtheorem{definition}{Definition}
\newtheorem{proposition}{Proposition}
\newtheorem{example}{Example}
\newtheorem{revisit}{Revisiting Example}

\newtheorem{appendixLemma}{Lemma}

\newtheorem{propAppendix}{Proposition}

\begin{document}
\bstctlcite{IEEEexample:BSTcontrol}
       \title{From Discrete to Continuous Highest-earning Imitation Dynamics
       %:\\ Discrete Fluctuations  Almost Surely Vanish with Population Size
       }

   \author{
       Azadeh Aghaeeyan 
       and~Pouria~Ramazi% <-this % stops a space
  \thanks{A. Aghaeeyan is with the Department of Mathematics and Statistics, Brock University, Canada (e-mail: a.aghaeeyan@gmail.com). P. Ramazi is with the Departments of Mathematics and Biological Sciences, the University of Calgary, Canada (email: pouria.ramazi@ucalgary.ca)
  This paper was presented in part at the  63rd IEEE Conference on Decision and Control, Milan, Italy, 2024 \cite{aghaeeyan2024discreteCDC}.
 Copyright may be transferred without notice
}%
 }

%\markboth{IEEE TRANSACTIONS ON 
%}{A \MakeLowercase{\textit{et al.}}: On convergence of continuous dynamics}

% ==========================================================
\maketitle

% === ABSTRACT ====================================================================
% Lessons learned from continuous-time  domains: No fluctuations persist in discrete time/Discrete fluctuations do not persist or scale as fast as population size
% =================================================================================
\begin{abstract}
Imitating the highest earners is a common decision-making heuristic, but in finite populations it can generate persistent fluctuations between strategies. This paper studies whether such fluctuations persist as population size grows in heterogeneous two-strategy populations. We show that the Markov chains describing the discrete imitation dynamics form generalized stochastic approximation processes for a good upper semicontinuous differential inclusion, which defines the associated mean dynamics. We prove that these mean dynamics always converge to equilibria. Using stochastic approximation results, we then show that the amplitudes of fluctuations in the population proportions of the two strategies vanish almost surely as population size tends to infinity. Thus, in well-mixed large populations, highest-earning imitation is unlikely to produce large-scale perpetual fluctuations. 
\end{abstract}

% === KEYWORDS ====================================================================
% =================================================================================

% \fi
\IEEEpeerreviewmaketitle

% ====================================================================

% === I. INTRODUCTION =============================================================
% =================================================================================
\section{Introduction}
Individuals regularly face decisions, such as whether to receive a flu shot or purchase a new product.
In these contexts, it is often assumed or observed that individuals chiefly behave as either \emph{best-responders} or \emph{imitators} \cite{tanny1988innovators, van2007new,bauch2005imitation}.
Best-responders choose the option that maximizes their immediate benefit, while imitators pay attention to others' decisions and how satisfied they are with those decisions \cite{10453658,10172285}.

In the long term, a finite population of best responders who benefit from majority agreement equilibrates \cite{ramazi2020convergence}, whereas one benefiting from majority disagreement either equilibrates or fluctuates between two adjacent states \cite{ramazi2017asynchronous,grabisch2020anti}.
 A mixed population of these two types may either equilibrate or fluctuate \cite{roohi}.
The same is true for a population imitating the highest earners
 \cite{FU2024111354}.
 Obtaining the conditions for the existence of fluctuations and characterizing them appear to be challenging.

The \emph{mean dynamics} associated with these population dynamics can be considered an approximation that simplifies analysis \cite{sandholm2010population}.
The mean dynamics of best-response rules are typically differential inclusions \cite{golman2010basins,bestresponsePotential,berger2007two,theodorakopoulos2012selfish}, while noisy best responses yield smooth logit dynamics \cite{cianfanelli2025stability}.
As for the imitation update rules, Lipschitz continuous switching rates between available decisions result in  Lipschitz continuous  mean dynamics including \emph{replicator dynamics} 
 \cite{comoImitation, replicator, replicator2,nogales2020replicator, 10842053}.
 The connection between the behavior of the discrete population dynamics and that of their associated mean dynamics
\cite{benaim2005stochastic,benaim1998recursive, benaim2003deterministic, benaim2009mean, roth2013stochastic} was extensively studied.

Recently, we linked finite-population binary best-response dynamics to their mean dynamics and showed that the amplitudes of perpetual fluctuations in strategy proportions almost surely vanish with population size \cite{aghaeeyan2023discrete}. \tb{Whether the same conclusion holds for highest-earning imitation remains open and is especially relevant because imitators were shown to have a lower convergence tendency than best responders \cite{ramazi2022lower}.
Highest-earning imitation is a heterogeneous-population version of the ``imitate-the-best'' rule in evolutionary game theory \cite{nowakMay1992}. 
It also models success- or payoff-biased social learning, where agents rely on others' observed success as complete payoff information may be unavailable. 
Such learning has been reported in humans \cite{henrich2001evolution} and animals \cite{BarrettEtAl2017}.}

Our contribution is threefold. First, we derive the mean dynamics associated with the discrete highest-earning imitation population dynamics and show that they are a good upper semicontinuous differential inclusion rather than an ordinary differential equation. 
Second, we prove that the family of Markov chains describing the discrete dynamics, indexed by population size, is a GSAP for these mean dynamics--\textbf{\Cref{lem_GSAP}}. 
Using the results of Roth and Sandholm \cite{roth2013stochastic}, we show that the mean dynamics approximate the evolution of the discrete imitation population dynamics when the population size approaches infinity, 
both on finite horizon--\textbf{\Cref{prop_finitie_horizon}}--and in the long term--\textbf{\Cref{thm:2}} and \textbf{\Cref{cor_fluctuationsDoNotScaleWithN_discretePopulationDynamics_2}}.
Third, using the convex-hull structure at boundary and discontinuity points, we characterize the equilibria and prove that all trajectories of the mean dynamics converge to them--\textbf{\Cref{prop:birkhoffcenter}}. Consequently, the fluctuation amplitudes reported in \cite{FU2024111354} vanish almost surely at the population-proportion scale as population size tends to infinity.
The results suggest that in a well-mixed population where individuals imitate the highest earners, fluctuations in the population proportions of the two strategies are appreciable for smaller population sizes.

\tb{
Compared with the best-response dynamics in \cite{aghaeeyan2023discrete}, the highest-earning imitation rule studied here leads to a more involved large-population limit. Although we use the same stochastic-approximation framework of \cite{roth2013stochastic}, the population proportion of $\1$-players no longer fully specifies the population state, and the reduction to an abstract scalar dynamics must be justified separately, especially for trajectories starting on the boundary. 
We therefore formulate the mean dynamics as a differential inclusion whose values on the boundary and on the set where the maximum $\1$-utility equals the maximum $\2$-utility are given by the convex hull of two vector fields.
This model-specific construction is not supplied by \cite{roth2013stochastic} or by \cite{FU2024111354}, which analyzes the finite-population imitation dynamics. 
Our analysis extends \cite{FU2024111354} to the large-population regime and shows that the finite-population fluctuations reported there vanish at the population-proportion scale. 
Also, whereas the preliminary version \cite{aghaeeyan2024discreteCDC} treated only diagonal coordination and antidiagonal anticoordination payoff matrices under an additional simplifying assumption, the present paper generalizes the framework by allowing utilities to be polynomial functions of the proportion of $\1$-players.
}

%\subsection*{Notations}
\textbf{Notations.}
Boldface letters denote vectors.
Sets are denoted by calligraphic fonts $\mathcal{X}$.
% Consider a sequence of scalars $\langle k \rangle = k_0, k_1, \ldots$.
By $\langle x_{k}\rangle_{k = 0}^{\infty}$, we  mean a sequence of variables $x_{0}, x_{1}, x_{2},\ldots$.
The floor function is denoted by $\lfloor x \rfloor$.
% The notation $\mb{e}_i$ denotes the $i^{\text{th}}$ standard basis vector.
The notation $\vert x \vert$ refers to the norm-1 of vector  $x$.
The $i^\text{th}$ standard basis vector in $\mathbb{R}^n$ is denoted by $\mb{e}_i$.
A set-valued map ${\bm {\mathcal{V}}}(\x)$ from $\mathbb{R}^n$ to $\mathbb{R}^n$ is denoted by notation ${\bm {\mathcal{V}}}: \mathbb{R}^n\rightrightarrows \mathbb{R}^n$.
% The notation $[a,b]-c$ implies $[a-c,b-c]$.
 % The set of all subsets of the set $\X$ is shown by the notation $2^{\X}$.
 The interior (resp. boundary) of a set $\X$ is denoted by $\mathrm{int}(\X)$ (resp. $\partial \X$).
The set $\{ 1,2,\ldots,k\}$ for a positive integer $k$ is represented by $[k]=\{1,2,\ldots,k\}$.
% The support of a random variable ${X}$ is denoted by $\mathcal{R}_{{X}}$.
The \tb{indicator} function $\mt{1}(\cdot)$ equals one for a positive argument and negative infinity otherwise.
The notation $\bm 1$ (resp. $\bm 0$)
refers to a vector with all elements equal to 1 (resp. 0) with an appropriate dimension.
The set of $m$-dimensional vectors whose components are positive rational numbers with a denominator dividing $n$, are denoted by   $\frac{1}{n}\mathbb{Z}^m$.
The smallest closed convex set containing  two vectors $\x_1,\x_2 \in \mathbb{R}^n$ is denoted by
$\mathrm{Conv}(\x_1, \x_2)$.
\section{Problem Formulation} \label{sec:problemFormulation}
Consider a well-mixed population of $\N$ agents choosing, back and forth, between two strategies $\1$ or $\2$ over time $t$ indexed by $k \in \mathbb{Z}_{\geq 0}$.
Each agent plays against the entire population, including herself, and, accordingly, obtains  \textit{utility}.
\tb{The utility of each agent is a polynomial function of} the population proportion of $\1-$players, i.e., the ratio of the number of $\1$-players to the population size $\N$.
Agents sharing the same utility function
form  a \emph{type}, and there are altogether $\p$ types labeled by $1,2,\ldots, \p$.
\tb{At time index $k$, when}
the population proportion of $\1$-players is $x^\N(k)$,
the utility of a type-$i$ agent, $i \in [\p]$, 
\tb{from playing $\1$} is denoted by $u^\1_i(x^\N(k))$, where
\begin{equation*} \label{eq:UAi}
       \tb{ u^\1_i(x) = \displaystyle \sum_{l = 0}^\texttt{d} a_{i,l} x^l}.
\end{equation*}
\tb{Here, $\texttt{d}$ is a positive integer, and $a_{i,l}$, $l \in [\texttt{d}]$, is a real number.
Similarly, the utility of agents of type $i$, $i \in [\p]$, 
\tb{from playing $\2$} is denoted by $ u^\2_i(x^\N(k))$ where}
\begin{equation*} \label{eq:UBi}
       \tb{ u^\2_i(x) = \displaystyle \sum_{l = 0}^\texttt{d} b_{i,l} x^l},
\end{equation*}
\tb{
and  $b_{i,l}$, $l \in [\texttt{d}]$, is a real number.
The utility functions considered here are quite general, where the matrix-game model
is a special case obtained by setting the coefficients $a_{i,l}$ and $b_{i,l}$ equal to zero for $l >1$ \cite{aghaeeyan2024discreteCDC}.}
Henceforth, the utility of $\1$-players (resp. $\2$-players) of type $i$, $i \in [\p]$, is referred to as \emph{type-$i$'s $\1$-utility} (resp. \emph{type-$i$'s $\2$-utility}).

The distribution of the population proportions over the total $\p$ types is shown by
$\bm{\rho} = ({\rho}_1, \ldots, {\rho}_{\p})^\top$ where  ${\rho}_p$ denotes the number of agents in type $p$ divided by the population size $\N$, i.e., the  population proportion of type $p$.
Define the \emph{population state} at time index $k$ as the distribution of $\1$-players over the $\p$ types, i.e.,
$
    \x^{\mathsf{N}} (k)
    = ({x}^{\mathsf{N}}_1(k),
    \ldots, 
    {x}^{\mathsf{N}}_{{\p}}(k)
    )^\top,
$
where ${x}^{\N}_p(k)$, $p \in [\p],$ equals the proportion of $\1$-players of type $p$.
The state space then equals 
$\bm{\mathcal{X}}_{s} \cap \frac{1}{\mathsf{N}}\mathbb{Z}^{{\p}}$
where 
$\bm{\mathcal{X}}_{s} = \prod_{j=1}^{{\p}}[0,{\rho}_j].$
%The population proportion of cooperators at index $k$, $x^\N(k)$, is indeed $\sum_{i=1}^\p x^\N_i(k).$

The activation sequence is \emph{asynchronous}\tb{,} that is, at each time index $k$, exactly one agent revisits her strategy according to the \textit{(highest-earning) imitation} update rule, which involves switching to the strategy played by the current highest-earning players. 
If both strategies are played by the current highest-earning players, the active agent chooses $\1$ \tb{(See \Cref{rem_tie} for some other tie-breaking rules)}.
More specifically, the \emph{preferred strategy} of the population at population state $\x^\N$ is defined by
\begin{equation} \label{eq:S}
\begin{aligned}
  \mathtt{s}(\x^\N) =  
\begin{cases}
 \1,  & \text{ if } u^\1(\x^\N) \geq   u^\2(\x^\N),\\
 \2,   & \text{ if }   u^\2(\x^\N) >  u^\1(\x^\N), \\
\end{cases}
\end{aligned}
\end{equation}
where
\begin{align} 
    u^\1(\x^\N) &= \max_{i\in[\p]} u^\1_i(x^\N)\mt{1}(x^\N_i), \label{eq_max_C_utility}\\
    u^\2(\x^\N) &= \max_{i\in[\p]} u^\2_i(x^\N)\mt{1}(\rho_i - x^\N_i). \label{eq_max_D_utility}
\end{align}
The indicator function $\mt{1}(\cdot)$ captures whether there are agents of a particular type playing a specific strategy.

The asynchronous activation sequence is denoted by $\langle A_k\rangle_{k = 0}^{\infty}$, where $A_k$ is the active agent at time index $k$, \tb{independently drawn}, and follows the uniform random  distribution, i.e.,  $\mathbb{P}(A_k = i) = \frac{1}{\N}$ if $i \in [\N]$ and zero otherwise.

The evolution of the population state over time defines the population dynamics or \textbf{\emph{discrete imitation population dynamics}}.
The population state, imitation dynamics, and the activation sequence fully describe the dynamics.

It was shown that a finite population of agents imitating the highest earners
may undergo perpetual fluctuations in the long term \cite{FU2024111354}, where the proportion of $\1$-players never converges to a fixed value.
It is, however, unclear whether the amplitude of the fluctuations grows with population size.
\begin{example} \label{exampleFinite} 
    Consider a population of $\N$ agents stratified into \tb{six} types with the distribution of population proportions \tb{$\bm \rho = (0.2, 0.05, 0.25, 0.125, 0.125, 0.25)^\top$} and \tb{utilities}
%     \tb{
%     \begin{equation*}
%   \scalebox{0.87}{$  \pi_1 \!=\!
% \begin{pmatrix}
%    1.12 & 1.87\\
%    -0.48 & 1.9 \\
% \end{pmatrix}, 
%  \pi_2\! =\!
% \begin{pmatrix}
%    2.43 & -0.33\\
%    0.82 & 0.47 \\
% \end{pmatrix}, 
%  \pi_3\! =\!
% \begin{pmatrix}
%    0.17 & -0.9\\
%    -0.67 & 0.62 \\
% \end{pmatrix},
% $}
% \end{equation*} }
% \tb{
%  \begin{equation*}
%   \scalebox{0.87}{$
% \pi_4 \!=\!
% \begin{pmatrix}
%    1.39& 1.16\\
%    1.98 & 0.98 \\
% \end{pmatrix}, 
%  \pi_5\! =\!
% \begin{pmatrix}
%    1.75 & 0.17\\
%    0.66 & 0.32 \\
% \end{pmatrix}, 
%  \pi_6\! =\!
% \begin{pmatrix}
%    2.09 & 0.77\\
%    -1.03 &  1.27\\
% \end{pmatrix}.
% $}\end{equation*}}%
\tb{
\[
\begin{array}{ll}
u^{\1}_1=-0.75x^\N+1.87, & u^{\2}_1=-2.38x^\N+1.90, \\
u^{\1}_2=2.76x^\N-0.33,  & u^{\2}_2=0.35x^\N+0.47, \\
u^{\1}_3=1.07x^\N-0.90,  & u^{\2}_3=-1.29x^\N+0.62, \\
u^{\1}_4=0.23x^\N+1.16,  & u^{\2}_4=1.00x^\N+0.98, \\
u^{\1}_5=1.58x^\N+0.17,  & u^{\2}_5=0.34x^\N+0.32, \\
u^{\1}_6=1.32x^\N+0.77,  & u^{\2}_6=-2.30x^\N+1.27.
\end{array}
\]
}
    For each value of \tb{$\N=40, 80, 160,$ and $320,$} we simulated the population dynamics \tb{$100$} times starting from the initial condition \tb{$\x^\N_0 = (4/40, 1/40, 2/40, 1/40, 3/40, 3/40)^\top$} with different random activation sequences. 
    We ran the simulations for \tb{$50\times\N$} steps and recorded
    the minimum and maximum of the population proportions of $\1$-players during the last \tb{$5\N$} steps.
  %  For each population size $\N$, we reported these values in
    These values  approach each other as the population size increases (Figure \ref{fig:finitepop1}).
\end{example}
\tb{\Cref{example2Finite} presents a similar simulation study for a population consisting of $10$ types.}
\begin{example} \label{example2Finite}
   \tb{There is a population of $\N$ agents stratified into ten types with the distribution of population proportions \tb{$\bm \rho = (17, 10, 14, 14,  4, 10, 17,  9, 20, 16)^\top/131$} and utilities}
%     \tb{
% \begin{equation*}
%   \scalebox{0.87}{$
% \pi_1 \!=\!
% \begin{pmatrix}
%   0.73  & -0.35\\
%   -0.56 & 1.65
% \end{pmatrix}, 
% \pi_2\! =\!
% \begin{pmatrix}
%   1.58 & 1.89\\
%   0.20 & 1.06
% \end{pmatrix}, 
% \pi_3\! =\!
% \begin{pmatrix}
%   0.02 & 1.05\\
%   1.62 & 1.27
% \end{pmatrix}.
% $}
% \end{equation*}
% }
% \tb{
% \begin{equation*}
%   \scalebox{0.87}{$
% \pi_4 \!=\!
% \begin{pmatrix}
%   0.21 & -1.42\\
%   -1.33 & 0.10
% \end{pmatrix}, 
% \pi_5\! =\!
% \begin{pmatrix}
%   1.82 & 0.14\\
%   1.93 & 1.14
% \end{pmatrix}, 
% \pi_6\! =\!
% \begin{pmatrix}
%   -0.09 & 0.28\\
%   -1.15 & 1.93
% \end{pmatrix}.
% $}
% \end{equation*}}
% \tb{
% \begin{equation*}
%   \scalebox{0.87}{$
% \pi_7 \!=\!
% \begin{pmatrix}
%   0.05 & -1.25\\
%   1.88 & 0.70
% \end{pmatrix}, 
% \pi_8\! =\!
% \begin{pmatrix}
%   -0.37 & 1.99\\
%   1.08 & 2.40
% \end{pmatrix}, 
% \pi_9\! =\!
% \begin{pmatrix}
%   0.75 & -0.78\\
%   1.82 & 1.22
% \end{pmatrix}.
% $}
% \end{equation*}}
% \tb{
% \begin{equation*}
%   \scalebox{0.87}{$
% \pi_{10} \!=\!
% \begin{pmatrix}
%   2.05 & -0.96\\
%   0.12 & 1.93
% \end{pmatrix}.
% $}
% \end{equation*}}
\tb{
\[
\begin{array}{ll}
u^{\1}_1=1.08x^\N-0.35,      & u^{\2}_1=-2.21x^\N+1.65, \\
u^{\1}_2=-0.31x^\N+1.89,     & u^{\2}_2=-0.86x^\N+1.06, \\
u^{\1}_3=-1.03x^\N+1.05,     & u^{\2}_3=0.35x^\N+1.27, \\
u^{\1}_4=1.63x^\N-1.42,      & u^{\2}_4=-1.43x^\N+0.10, \\
u^{\1}_5=1.68x^\N+0.14,      & u^{\2}_5=0.79x^\N+1.14, \\
u^{\1}_6=-0.37x^\N+0.28,     & u^{\2}_6=-3.08x^\N+1.93, \\
u^{\1}_7=1.30x^\N-1.25,      & u^{\2}_7=1.18x^\N+0.70, \\
u^{\1}_8=-2.36x^\N+1.99,     & u^{\2}_8=-1.32x^\N+2.40, \\
u^{\1}_9=1.53x^\N-0.78,      & u^{\2}_9=0.60x^\N+1.22, \\
u^{\1}_{10}=3.01x^\N-0.96,   & u^{\2}_{10}=-1.81x^\N+1.93.
\end{array}
\]
}
   \tb{For each value of $\smash{\N=131, 262, 524,}$ and $1048$ we simulated  the population starting from the initial condition $\x^\N_0 = (16, 6, 4, 8, 3, 3, 9, 3, 15, 1)^\top/131$ with different random activation sequences and  the same set-up as in \Cref{exampleFinite}.
  %  For each population size $\N$, we reported these values in
    The bottom panel of \Cref{fig:finitepop1} reports the minimum and maximum  population proportions of $\1$-players over the last $5\N$ steps.}
\end{example}

According to these two examples, the amplitudes of the fluctuations in the \emph{population proportion of $\1$-players} reduce for larger population sizes.
But is this observation valid only for this specific example, or does it hold for more general cases?
\begin{figure}
    \centering
    \includegraphics[width = 1\linewidth]{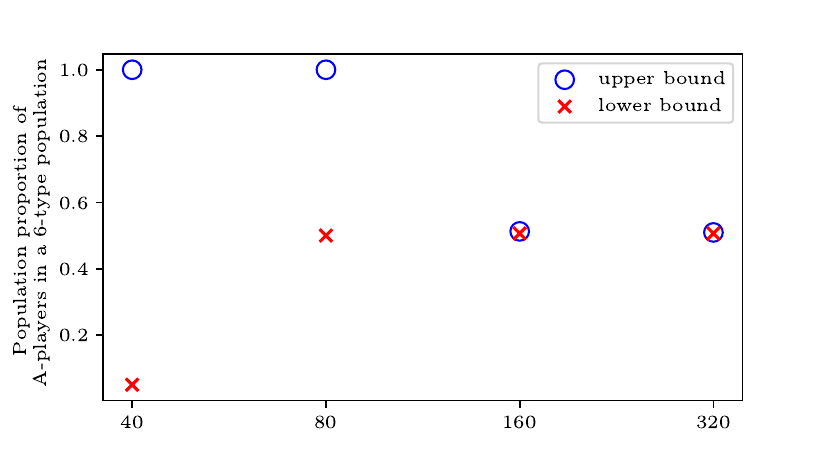}
        \includegraphics[width = 1\linewidth]{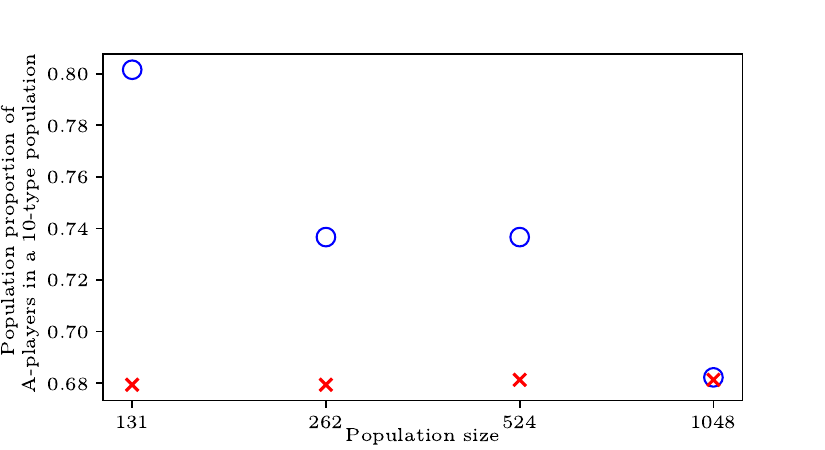}
    \caption{\textbf{Fluctuations in the population proportion of $\1$-players in the long-term for varying population sizes.} 
  \tb{  The top panel shows simulation results for a population consisting of 6 types (\Cref{exampleFinite}); the bottom panel shows simulation results for a population consisting of 10 types (\Cref{example2Finite}).}
     Each circle (resp. cross) denotes the maximum (resp. minimum) of the recorded population proportions of $\1$-players for each population size.
}
    \label{fig:finitepop1}
\end{figure}
Roth and Sandholm lay the foundations to answer this question \cite{roth2013stochastic}.
To leverage their result, we should first show that the discrete imitation population dynamics admit a Markov chain, and the families of these Markov chains, indexed by the population size, are GSAPs for a good upper semicontinuous differential inclusion--the associated mean dynamics.
Then, the question can be answered by investigating the asymptotic behavior of the mean dynamics.
We begin with some preliminaries.

\section{Background}
A set-valued map
$\bm{\V}\!:\!\bm{\mathcal{X}}\!\rightrightarrows\!\mathbb{R}^n$
defines a \emph{differential inclusion} on $\!\bm{\mathcal{X}}\!$ by $\!\dot{\x}\!\in\!\bm{\V}(\x)$.
An \emph{equilibrium} of the differential inclusion is a  point $\!\x^*\! \in\!\bm{\mathcal{X}}$ that satisfies $\!{\mb{0}}\!\in\!\bm{\mathcal{V}}(\x^*)\!$ \cite{cortes2008discontinuous}.
Assume that $\bm{\mathcal{X}}$ is a closed convex subset of $\mathbb{R}^n$ and for each $\x \in \bm{\mathcal{X}}$ we have
$\bm{\V}(\x) \subseteq \mathrm{T}_{\bm{\mathcal{X}}}(\x)$, where $\mathrm{T}_{\bm{\mathcal{X}}}(\x)$ is the \emph{tangent cone} of the set $\bm{\mathcal{X}}$ at $\x$ i.e., $\text{cl}\big(\{\bm z \in \mathbb{R}^n : \bm z = \alpha (\y -\x) \text{ for some }\y\in \bm{\mathcal{X}}  \text{ and } \text{some } \alpha \geq 0 \}\big)$.
Then the differential inclusion $\dot{\x}\in \bm{\V}(\x)$ is
\emph{good upper semicontinuous} if for any $\x \in \bm{\X}$ it is  nonempty, convex,  bounded, and
 its \emph{graph}, i.e., $\{(\x,\y) \mid \y \in \bm{\mathcal{V}}(\x) \}$, is closed \cite[6.A]{sandholm2010population}.
The \emph{basin of attraction} of the set $\!\bm{\mathcal{M}}\!\subseteq\!\bm{\mathcal{X}}\!$  under differential inclusion 
 $\bm{\V}$
with domain $\bm{\mathcal{X}}$ is defined as 
the union of all sets $\bm{\mathcal{U}} \subseteq \bm{\mathcal{X}}$ from which $\bm{\mathcal{M}}$ is attractive \cite{MAYHEW20111045}, i.e.,
  for each solution $\x(t)$ with $\!\x(0)\!\in\!\bm{\mathcal{U}}$ and each open $\epsilon$-neighborhood  of $\bm{\mathcal{M}}$, there exists some time $T>0$ such that for all $t \geq T$ the solution $\x(t)$ is in an $\epsilon$-neighborhood of $\bm{\mathcal{M}}$  \cite{MAYHEW20111045}.
      For a differential inclusion $\dot{\x} \in \bm{\mathcal{V}}(\x)$, defined over the compact and convex state space $\bm{\mathcal{X}}$, let $\bm{\mathcal{T}}_{\x_0}$ be the set of solutions starting from initial condition $\x_0$. 
      The set $\bm{\mathcal{T}}_{\x_0}$ is a subset of the space of continuous maps from $\mathbb{R}_+$ to $\bm{\mathcal{X}}.$
    We define the set-valued dynamical system induced by 
    $\dot{\x} \in \bm{\mathcal{V}}(\x)$ by 
     $\bm{\Phi}:\mathbb{R}_+ \times \bm{\mathcal{X}} \rightrightarrows \mathbb{R}^n$ where
    $\bm{\Phi}_t(\x_0) = \{\bm{x}(t) \in \bm{\mathcal{T}}_{\x_0} \}$.
     The set of recurrent points of $\bm{\Phi}$ are defined as
        $ \bm{\mathcal{R}}_{\bm{\Phi}} 
        =  \{\x_0 \lvert 
        \x_0 \in \bm{\mathcal{L}}(\x_0)  \},
    $
    where $\bm{\mathcal{L}}(\x_0)$ is the \emph{limit set} of point  $\x_0$ defined by $\bigcup_{\y \in \bm{\mathcal{T}}_{\x_0}}  \bigcap_{t \geq 0} \text{cl}(\y[t, \infty])$.
    The \emph{Birkhoff center} of $\bm{\Phi}$ is defined as the closure of $\bm{\mathcal{R}}_{\bm{\Phi}}$.
    
\begin{definition}[Adapted from \cite{roth2013stochastic}] \label{defGSAP} 
 Let $\dot{\x} \in \bm{\mathcal{V}}(\x)$ be  a good upper semicontinuous differential inclusion  satisfying 
    $\bm{\mathcal{V}}(\x) \subseteq \mathrm{T}_{\bm{\mathcal{X}}}(\x)$ for any $\x \in \bm{\mathcal{X}}$, where
     $\bm{\mathcal{X}}$  is a convex and compact state space.
     Consider a sequence of positive values $\epsilon$ approaching zero.
  Let $\mathbi{U}^{\epsilon} = \langle \mathbi{U}^{\epsilon}_k\rangle_{k = 0}^{\infty} $  be 
    a sequence of $\mathbb{R}^n$-valued random variables and $\langle\bm{\mathcal{V}}^{\epsilon}\rangle$ be 
   a family of set-valued maps on $\mathbb{R}^n$. 
     The family $ \langle \mb{X}^{\epsilon}_k \rangle_{k = 0}^{\infty}$ is called a family of \emph{generalized stochastic approximation processes} (or \emph{GSAPs}) for $\dot{\x} \in \bm{\mathcal{V}}(\x)$ if the following are met:
    \begin{enumerate}
        \item $\mb{X}^{\epsilon}_k \in {\bm{\mathcal{X}}}$ for all $k \geq 0$, 
        \item $\mb{X}^{\epsilon}_{k+1} - \mb{X}^{\epsilon}_k - \epsilon \mathbi{U}^{\epsilon}_{k+1} \in \epsilon \bm{\mathcal{V}}^{\epsilon}(\mb{X}^{\epsilon}_k)$,
        \item
        $\forall \delta >0 \exists \epsilon_0>0,$ 
         $\forall \epsilon \leq \epsilon_0 \forall \x \in \bm{\mathcal{X}}$
        \begin{equation} \label{eq_3cnd_GSAP}
        \bm{\mathcal{V}}^{\epsilon}\!(\x)\!\subset\!\{\!\bm{z} \!\in\! \mathbb{R}^n\!\mid\!\exists \y,\!\vert\x-\y\vert< \delta, \!\inf_{\bm{v}\in \bm{\mathcal{V}}(\y)}\! \vert \!\bm{z} - \bm{v}\vert \!< \!\delta \!\},
        \end{equation}
        \item  For all $ T>0$  and all $\alpha>0$, we have
        $
        \lim_{\epsilon \rightarrow 0} \mathbb{P} 
        \left[ \max_{k \leq \frac{T}{\epsilon}} \left \vert 
        \textstyle\sum_{i = 1}^{k} \epsilon \mathbi{U}^{\epsilon}_{i} \right \vert > \alpha \mid \mb{X}^{\epsilon}_0 = \x \right] = 0
        $
        uniformly in $\x \in \bm{\mathcal{X}}$.
    \end{enumerate}
    \end{definition}

For a discrete-time process
$\langle  \mb{X}^{\frac{1}{\N}}_k\rangle_{k = 0}^{\infty},$
 we define an associated
\emph{interpolated process} $\bar{\bm{X}}^\frac{1}{\N}$ running in continuous time as follows
$
\bar{\bm{X}}^\frac{1}{\N}(t) =  \mb{X}^{\frac{1}{\N}}_{l(t)} + \big(\N t - l(t) \big) \big(\mb{X}^{\frac{1}{\N}}_{l(t)+1} - \mb{X}^{\frac{1}{\N}}_{l(t)}\big),
$
where $l(t) = \lfloor t\N \rfloor$.
\begin{theorem} [Adapted from \cite{roth2013stochastic} ]\label{thm:shortTermSandholm}
    Let  $ \langle \mb{X}^{\epsilon}_k \rangle_{k = 0}^{\infty}$  be a family of GSAPs for a differential inclusion 
    $\dot{\x} \in \bm{\mathcal{V}}(\x)$, defined in \Cref{defGSAP}, 
    and let $\bm{\mathcal{T}}_{\Phi}$ be the set of all solutions to the differential inclusion.
    Then for any $T>0$ and any $\alpha >0$, we have
    $$
    \lim_{\frac{1}{\N} \to 0} \mathbb{P}\! \left[ \!\inf_{
    \displaystyle\x(t) \in \bm{\mathcal{T}}_{\Phi}} \sup_{0 \leq s \leq T} \left \vert \!\bar{\bm{X}}^{\frac{1}{\N}}(s) - \x(s) \right \vert \! \geq \alpha \!\mid \!\bar{\bm{X}}^{\frac{1}{\N}}(0) \!= \!\x_0 \right]\! = \!0
    $$
    uniformly in $\x_0 \in {\bm{\mathcal{X}}}$.
\end{theorem}
\tb{An intersection point $x_0$ of two univariate polynomials $f(x_0)$ and $g(x_0)$ is \emph{transverse} if $\dot{f}(x_0) - \dot{g}(x_0) \neq 0$. }
\section{Markov chain and mean dynamics}

To obtain the Markov chain associated with the discrete imitation population dynamics, we first write the dynamics in a compact form.
Define $s(\x^\N)$ as a function 
that returns $1$ (resp. $2$) if $\1$ (resp. $\2$) is
the preferred strategy of the population at state $\x^\N$, or, equivalently,
\begin{align} \label{eq:u}
{s}(\x^{\N})  =
    \begin{cases}
   1, & \text{if }  u^\1(\x^\N) \geq   u^\2(\x^\N),    \\
   2, & \text{if }  u^\1(\x^\N) <   u^\2(\x^\N).
    \end{cases}
\end{align}
\begin{proposition} \label{def:discreteDynamics}
   The 
   discrete imitation population dynamics 
  correspond to the dynamics described by
   following discrete time stochastic equation for $k\in\mathbb{Z}_{\geq0}$:
   \begin{align}\label{eq:discretePopulationDynamics}
        \x^{\mathsf{N}}(k+1)  &= \x^{\mathsf{N}}(k) + \frac{1}{\mathsf{N}} \big({S}_k-s(\x^{\N}(k))\big)\mb{e}_{{P}_k},
    \end{align}
   where ${P}_k$ and ${S}_k$ are random variables with
    distributions $\smash{\mathbb{P}[{P}_k  = p] = \rho_p, \mathbb{P} [{S}_k =1 \vert P_k = p] = x^{\N}_p/\rho_p}$, 
    and 
    $\mathbb{P} [{S}_k \! =\! 2 \vert P_k = p] = 1 - x^{\N}_p/\rho_p,$ for $p \in [\p]$,
    and
     supports
     $[\p]$ and $\{1, 2\}$, 
    respectively.
\end{proposition}
\begin{proof}
    According to the imitation population dynamics, at each  index $k$, exactly one agent becomes active with probability $1/\N$.
    Accordingly, the probability that the active agent belongs to type $p$ is $\rho_p$, represented by the random variable $P_k$.
    Given the definition of the population state, the probability that the active agent of type $p$ is playing $\1$ (resp. $\2$) is $x^\N/\rho_p$ (resp. $1-x^\N/\rho_p$), captured by the random variable  $S_k$.
    If the preferred strategy is the same as her current strategy, she keeps it, causing no changes in the population state; otherwise, she switches to the preferred strategy. 
    If the active agent switches from playing $\2$ (resp. $\1$) to playing $\1$ (resp. $\2$), 
     the number of $\1$-players in type $p$ will increase (resp. decrease) by $1$, thus changing the proportion of $\1$-players in type $p$  by $1/\N.$
    This process is  captured by  $\frac{1}{\N}\big(S_k - s(\x^\N)\big)\mb{e}_{P_k} $.
    Hence, this evolution is fully captured by the dynamics \eqref{eq:discretePopulationDynamics}.
\end{proof}

  In \Cref{def:discreteDynamics}, 
  the random variable ${P}_k$ denotes the type of the active agent at index $k$, and
  the random variable ${S}_k$  equals $1$ (resp. $2$) if  $\1$ (resp. $\2$)
  is the strategy of the active agent at time index $k$.
The corresponding Markov chain to the discrete imitation population dynamics \eqref{eq:discretePopulationDynamics} is defined as follows.
% By \emph{realization} of a Markov chain, we mean a specific sequence of the states starting from a specific initial condition and follows the transition probability of the Markov chain.
\begin{definition} \label{prop:markovchain}
 The \textbf{\emph{imitation population dynamics Markov chain}} 
 is defined as the Markov chain
 $\langle \mathbf{X}^{\frac{1}{\mathsf{N}}}_k\rangle_{k = 0}^{\infty}$  with 
 the state space $\bm{\mathcal{X}}_{s} \cap \frac{1}{\mathsf{N}}\mathbb{Z}^{{\p}}$, the initial state $\mathbf{X}^{\frac{1}{\N}}_0 = \x^{\N}(0),$ and the
 transition probabilities
 \small{
    \begin{align}
        & \text{Pr}_{\x^{\N},\y^{\N}} = \label{eq:markov}\\ 
    &\begin{cases} 
        ({\rho}_p - x^{\mathsf{N}}_p)(2-{s}(\x^{\mathsf{N}})), &\hspace{-3pt}\text{if }\exists p(\y^{\N} = \frac{1}{\mathsf{N}}\mb{e}_p + \x^{\N}), \nonumber \\
         x^{\mathsf{N}}_p ({s}(\x^{\mathsf{N}})-1), &\hspace{-3pt}\text{if } \exists p(\y^{\N} = -\frac{1}{\mathsf{N}}\mb{e}_p + \x^{\N}),
    \\
        1- \Big(\displaystyle\sum_{p=1}^{\p} ({\rho}_p - x^{\mathsf{N}}_p)(2-{s}(\x^{\mathsf{N}}))&\hspace{-3pt}\text{if } \y^{\N} =  \x^{\N},\\
        \hspace{15pt}  
        \quad  +  x^{\mathsf{N}}_p ({s}(\x^{\mathsf{N}})-1) \Big),
        & \\
        0, &\hspace{-3pt}\text{otherwise}. 
    \end{cases}   
\end{align}}
\end{definition}

The discrete imitation population dynamics and their Markov chain parallel those in \cite{aghaeeyan2023discrete}, but the switching rule $s(\cdot)$ differs, leading to different dynamics.

We increase the population size in a way that the population distribution, i.e.,  the vector of population proportions  $\bm{\rho}$, remains unchanged.
The elements of the sequence $\langle N\rangle_{N = \N_0}^{\infty}$ of the population size  should satisfy $N  \bm{\rho} \in \mathbb{Z}^{\p}$.   
    We henceforth assume that the population sizes satisfy this condition. 
   
The next step is to show that  $\langle \mb{X}^{\frac{1}
{\mathsf{N}}}_k\rangle_{k = 0}^{\infty} $, indexed by the population size $\N$,
is a GSAP for a good upper semicontinuous differential inclusion.
We claim that $\langle \mb{X}^{\frac{1}
{\mathsf{N}}}_k\rangle_{k = 0}^{\infty} $ is
a GSAP for the following differential inclusion.

\begin{definition}  \label{def_semicontinuousDynamics}
    The \textbf{\emph{continuous-time imitation population dynamics}} are defined by 
       $ \dot{\x}\in \bm{\mathcal{V}}(\x)$  with set-valued map
    $\bm{\mathcal{V}} :\bm{\mathcal{X}}_{s} \rightrightarrows \mathbb{R}^{\p}$ given by
    \begin{align} 
      \bm {\V}(\x&) =
    \label{eq:semicontinuous}
    \\
       &\begin{cases} 
    \{\bm \rho -\x\},  &  \hspace{-9pt} \text{ if } \x \notin \partial \bm{\mathcal{X}}_{s} \text{ and } u^\1(\x) >   u^\2(\x), \nonumber \\
        \{-\x\},  & \hspace{-9pt} \text{ if }  \x \notin \partial \bm{\mathcal{X}}_{s} \text{ and } u^\2(\x) >   u^\1(\x),\\
          \mathrm{Conv}(\bm \rho - \x, -\x),  &\hspace{-6pt} \text{otherwise\tb{.}}  \\
        \end{cases} 
    \end{align}   
\end{definition}

\begin{lemma}   \label{lem_GSAP}
    The collection of $ \langle \mb{X}^{\frac{1}{\mathsf{N}}}_k\rangle_{k = 0}^{\infty}$ is a GSAP for \eqref{eq:semicontinuous}.
\end{lemma}
\begin{proof}
   Given \Cref{lem_USC}, the differential inclusion \tb{\eqref{eq:semicontinuous} is upper semicontinuous.
   It is straightforward to show that $\bm{\mathcal{X}}_{s}$ is convex, so the first condition in \Cref{defGSAP} is satisfied.
   Let $\bm{\nu}^{\frac{1}{\N}}(\bm{x}^{\N})$ denote the expected increment per time unit of the Markov chain $\langle \mb{X}^{\frac{1}{\N}}_k\rangle_{k}$ at  $\bm{x}^{\N}$.
Since there are $\N$ transitions per unit time, we have
  $\bm{\nu}^{\frac{1}{\N}}(\bm{x}^{\N}) = \smash{\N \mathbb{E}[\mathbf{X}^{\frac{1}{\N}}_{k+1} - \mathbf{X}^{\frac{1}{\N}}_{k} \mid \mathbf{X}^{\frac{1}{\N}}_k = \x^{\N}]}$.
The $p^{\text{th}}$ element of $\bm{\nu}^{\frac{1}{\N}}$, $p \in [\p]$, which is the expected increment in the type $p$, equals the sum of the multiplication of each possible change with its probability.
In view of \eqref{eq:markov}, there are two possible changes:  an increase  or a decrease of size $\frac{1}{\N}$, yielding ${\nu}^{\frac{1}{\N}}_p({\x}^{\N}) =\N \big( \frac{1}{\N}(\rho_p - x^{\N}_p)(2-s(\x^{\N}))  -\frac{1}{\N} (x^{\N}_p (s(\x^{\N})-1))\big)$ and consequently, 
$
{\nu}^{\frac{1}{\N}}_p({\x}^{\N}) =\rho_p\big(2 - s(\x^{\N})\big)-x^{\N}_p 
$.
The second condition of \Cref{defGSAP} is satisfied by 
$ \mb{U}^{\frac{1}{\N}}_{k+1} =$ $\smash{
  \N(\mb{X}^{\frac{1}{\N}}_{k+1} - \mb{X}^{\frac{1}{\N}}_{k} -
 \mathbb{E}[\mathbf{X}^{\frac{1}{\N}}_{k+1} - \mathbf{X}^{\frac{1}{\N}}_{k} \mid \mathbf{X}^{\frac{1}{\N}}_k = \x^\N])}$
 resulting in
$\mb{X}^{\frac{1}{\N}}_{k+1} - \mb{X}^{\frac{1}{\N}}_{k} - {\frac{1}{\N}}\mb{U}^{\frac{1}{\N}}_{k+1} = {\frac{1}{\N}}\bm{\nu}^{\frac{1}{\N}}(\x)$.
As for the third condition, note that
${\bm \nu}^{\frac{1}{\N}}({\x})$ is a selection of the differential inclusion \eqref{eq:semicontinuous}. Hence, the condition \eqref{eq_3cnd_GSAP} is satisfied by
taking $\y=\x$.
As $ \mathbb{E}[\mb{U}^{\frac{1}{\N}}_{k+1} \vert \mathbf{X}^{\frac{1}{\N}}_k= \x^{\N}] =0$, $\mb{U}^{\frac{1}{\N}}$ is a Martingale difference sequence. In addition, 
 $ \mb{U}^{\frac{1}{\N}}$ is uniformly bounded by $\sqrt{\sum_{i=1}^\p (1 + \rho_i)^2}$. Hence, the last condition is satisfied thanks to \cite[Proposition 2.3]{roth2013stochastic}. This completes the proof.
   }
\end{proof}
\section{Finite horizon deterministic approximation}
Based on the following proposition, which is a straightforward application of \Cref{thm:shortTermSandholm}, the interpolated process of the imitation population dynamics Markov chain, denoted by $\bar{\bm{X}}^\frac{1}{\N}$, closely tracks one of the solutions of the continuous-time population dynamics \eqref{eq:semicontinuous} when population size approaches infinity.
\begin{proposition} \label{prop_finitie_horizon}
    For any $T>0$ and for any $\alpha >0$ we have
    $$
    \lim_{\frac{1}{\N} \to 0} \mathbb{P} \!\left[ \!\inf_{\x(t) \in \bm{\mathcal{T}}_{\Phi}} \!\sup_{0 \leq s \leq T} \vert \bar{\bm{X}}^{\frac{1}{\N}}(s) - \x(s) \vert \geq \alpha\! \mid \bar{\bm{X}}^{\frac{1}{\N}}(0) = \x_0 \right]\! = \!0
    $$
    uniformly in $\x_0 \in {\bm{\mathcal{X}}_{s}}$, where
    $\bar{\bm{X}}^\frac{1}{\N}$ is the interpolated process  of the imitation population dynamics Markov chain and
    $\bm{\mathcal{T}}_{\Phi}$ is the set of all solutions to the continuous-time population dynamics \eqref{eq:semicontinuous}.
\end{proposition}
\begin{proof}
    The collection $\langle \mb{X}^{\frac{1}{\N}} \rangle_{k=0}^\infty$  
     for the vanishing sequence $\langle \frac{1}{\N} \rangle_{N=\N_0}^\infty$ is a GSAP
     for the differential inclusion \eqref{eq:semicontinuous} (\Cref{lem_GSAP}), and given \Cref{thm:shortTermSandholm}, the proof is complete.
\end{proof}
\begin{revisit}
For the population in \Cref{exampleFinite}, when
$\x \notin \partial \bm{\mathcal{X}}_{s}$,  the equation $u^\1 (\x) = u^\2 (\x)$ is satisfied at \tb{$\bm 1^\top\x = 0.018, 0.509, 0.656.$}
Accordingly, the continuous-time population dynamics associated with the discrete imitation population dynamics in \Cref{exampleFinite} equal $\dot{\x} \in \bm {\V}(\x)$
where
\tb{
 \begin{small}
  \begin{align*}
  &\bm {\V}(\x) = \\
      &   \begin{cases} 
        \{- \x\}, &\hspace{-26pt} \text{if }  \x \notin \partial \bm{\mathcal{X}}_{s} \text{ and } \bm 1^\top \x \in (0, 0.018) \cup (0.509, 0.656),    \\
        \{\bm \rho - \x\}, &\hspace{-26pt}  \text{if }  \x \notin \partial \bm{\mathcal{X}}_{s} \text{ and }  \bm 1^\top \x  \in (0.018, 0.509) \cup (0.656, 1), \\
         \mathrm{Conv}(\bm \rho\! - \x,\! -\x), & \text{otherwise.} 
        \end{cases}
\end{align*}
 \end{small}
}

We simulated the above dynamics starting from the initial condition \tb{$(4/40, 1/40, 2/40, 1/40, 3/40, 3/40)^\top$} for one time unit.
The solution trajectory is depicted in the solid black curve in \Cref{fig:finiteHorizon}.
We also simulated one random realization of the imitation population dynamics Markov chain for each of the \tb{four} different population sizes starting from the same initial condition.
%The interpolated realizations are depicted in \Cref{fig:finiteHorizon}.
As can be seen in \Cref{fig:finiteHorizon}, the interpolated realizations and the solution of the continuous-time dynamics are getting closer for larger population sizes.
\end{revisit}
 \begin{figure*}
    \centering  \includegraphics[width=1\linewidth]{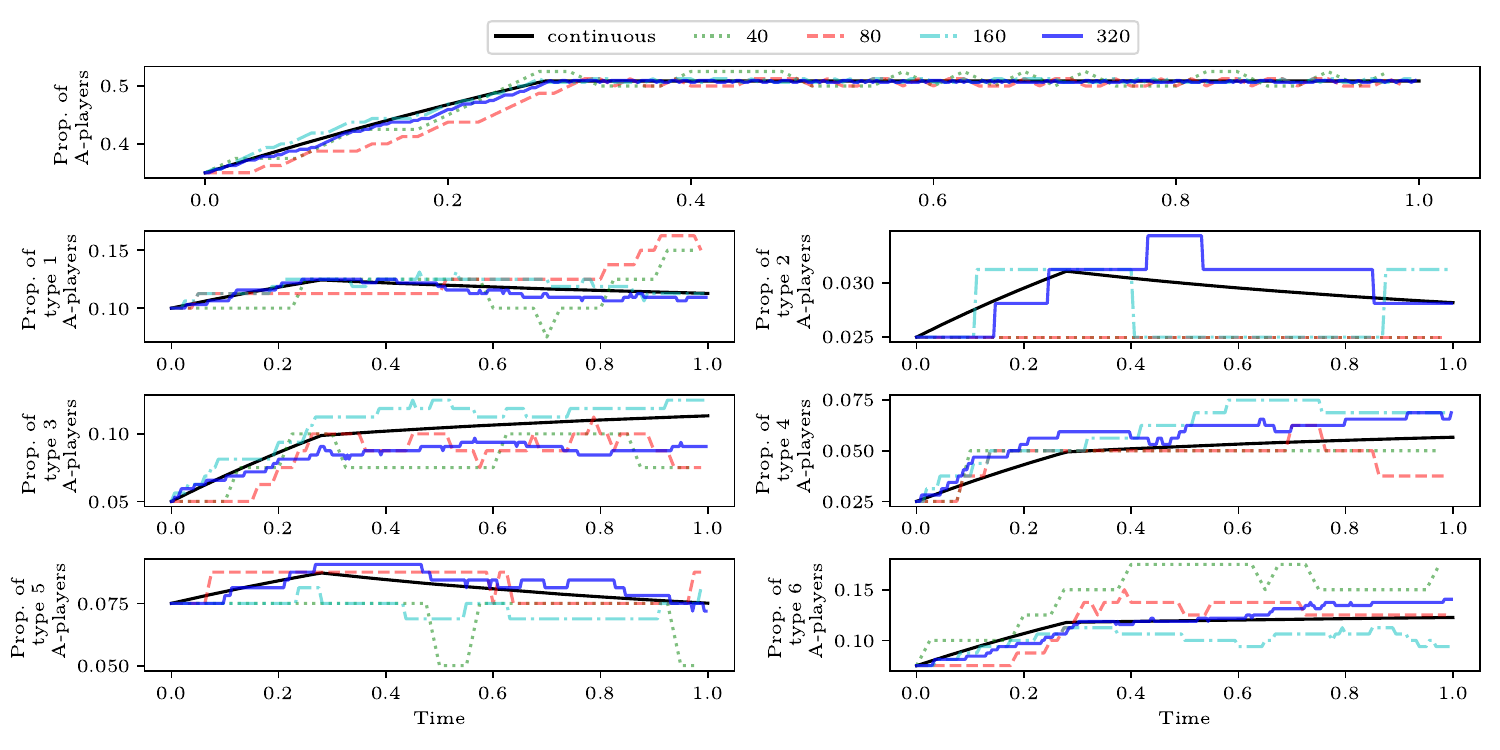}
  \caption{\tb{\textbf{For large population sizes, the trajectories of the proportion of $\1$-players obtained from the continuous-time population dynamics and from the interpolated discrete population dynamics closely match.}
The solid black curve shows the proportion of $\1$-players over time obtained from the continuous-time dynamics, while the other four curves represent the interpolated discrete population dynamics for different population sizes. 
}}
    \label{fig:finiteHorizon}
\end{figure*}

In the next section, we reveal the long-term behavior of the discrete imitation population dynamics, when the population size approaches infinity, by analyzing that of  the continuous-time imitation population dynamics.
\section{Asymptotic deterministic approximation}
% In this section, first we investigate the asymptotic behavior of the semicontinuous population dynamics and then, using the established results in the stochastic approximation theory, we link it to that of the discrete imitation population dynamics as the population size approaches infinity.
% \subsection{Abstract dynamics}
% In view of \eqref{eq:semicontinuous}, the semicontinuous dynamics are fully determined by the value of the proportion of cooperators, $x = \x^{\top}\bm 1$. 
% Hence, the analysis of the  evolution of the proportion of cooperators, $x$, sheds light on that of the semicontinuous population dynamics.
% Following \cite{aghaeeyan2023discrete},
In this section, using the established results in \cite{roth2013stochastic}, we connect the asymptotic behavior of the discrete imitation population dynamics, for population sizes approaching infinity, to that of the continuous-time population dynamics. 
\vspace{-11pt}
\subsection{Abstract Dynamics}
In view of \eqref{eq_max_C_utility} and \eqref{eq_max_D_utility}, if the population state is not at the boundary of the state space, i.e., $\x \notin \partial \bm{\mathcal{X}}_{s},$ the active case in \eqref{eq:semicontinuous} is fully determined by the population proportion of $\1$-players, i.e., $\bm 1^\top \x$.
Accordingly, analyzing the evolution of the population proportion of $\1$-players helps reveal that of the $\p$-dimensional population state.
\begin{definition}
The \emph{abstract dynamics} associated with the continuous-time population dynamics \eqref{eq:semicontinuous} are defined by
    $\dot{x} \in \mathcal{X}(x)$ with $\mathcal{X}:[0,1] \rightrightarrows [-1,1]$ given by
 \begin{equation} \label{eq:abstract}
  \mathcal{X}(x) = 
\begin{cases} 
        \{1-x\}, & \text{if } \displaystyle\max_{i\in[\p]} u^\1_i(x) >\displaystyle\max_{i\in[\p]} u^\2_i(x),  \\
        \{-x\}, & \text{if }  \displaystyle\max_{i\in[\p]} u^\1_i(x) <\displaystyle\max_{i\in[\p]} u^\2_i(x),  \\
         [-x, 1-x], & \text{if } \displaystyle\max_{i\in[\p]} u^\1_i(x) = \displaystyle\max_{i\in[\p]} u^\2_i(x)\tb{,}
        \end{cases}
\end{equation}
with the initial value  $x(0) = \bm 1^\top \x(0)$.
\end{definition}
% \begin{proof}
%     The proof is similar to that of \cite[Proposition 2]{aghaeeyan2023discrete} and is omitted.
% \end{proof}
We refer to $x$ in \eqref{eq:abstract} as the \emph{abstract state}.
We call 
$\max_{i\in[\p]} u^\1_i(x)$ (resp. $\max_{i\in[\p]} u^\2_i(x)$)
the \emph{maximum  $\1$-utility} (resp. \emph{maximum  $\2$-utility}) at the abstract state $x$.
\begin{assumption} \label{ass:2}
    % \begin{equation*}
    %     \{(i,j) \in [\p]\times [\p] \mid i \neq j, b_i = d_j, a_i-b_i+d_j - c_j =0 \} = \emptyset.
    % \end{equation*}
\tb{The utility functions are unique. 
In addition, at each intersection point of the maximum $\1$-utility and maximum $\2$-utility, on the unit interval, the maximum value is attained by a unique $\1$-utility and a unique $\2$-utility, and these two active utility functions intersect transversely.}
\end{assumption}
\Cref{ass:2} makes sure that the highest earning strategy at the intersection points of the maximum $\1$-utility and maximum $\2$-utility change.
We additionally make the following assumption.
\begin{assumption} \label{ass:making-A-chosen-strategy-at-1}
 The maximum  $\1$-utility at the abstract state $1$ (resp. $0$) is greater (resp. less) than the maximum  $\2$-utility.
\end{assumption}

In view of \Cref{ass:making-A-chosen-strategy-at-1},
 at the extreme points, the highest-earning strategy is the one played by the agents.

According to \cite[Proposition 1]{FU2024111354},
the equilibrium points of the discrete population dynamics are  the extreme points, i.e., $\bm 0$ and $\bm \rho$, and the intersections of the maximum $\1-$ and  $\2-$utilities.
The following lemma shows that the same characteristics hold for the equilibria of the abstract dynamics.
\begin{lemma} \label{lem:equilibriumOfAbstract}
    Under Assumptions \ref{ass:2} and \ref{ass:making-A-chosen-strategy-at-1},
    the extreme points, $0$ and $1$, and  the intersections of the maximum $\1-$ and $\2-$utilities are the equilibria of the
     abstract dynamics 
    \eqref{eq:abstract}.
   % , which are at least three and at most $2\p + 1$. 
\end{lemma}
\begin{proof}
At an equilibrium point of the abstract dynamics, $x_0$, we have $0 \in \mathcal{X}(x_0)$.
The state $x = 1$ (resp. $x = 0$) is the only candidate equilibrium point satisfying the first (resp. second) case in \eqref{eq:abstract}. 
Given \Cref{ass:making-A-chosen-strategy-at-1}, at $x = 1$ (resp. $x=0)$, we have $\max_{i\in[\p]} u^\1_i(x) >\max_{i\in[\p]} u^\2_i(x)$ (resp. $\max_{i\in[\p]} u^\2_i(x) >\max_{i\in[\p]} u^\1_i(x)$), and hence, the point $x=1$ (resp. $x=0$) is the equilibrium point of the abstract dynamics \eqref{eq:abstract}.
 As $0 \in [-x, 1-x]$ for all $x \in (0,1)$, the last case in \eqref{eq:abstract} indicates that any abstract state satisfying  
 $\max_{i\in[\p]} u^\1_i(x) =\max_{i\in[\p]} u^\2_i(x)$ is an equilibrium point of the abstract dynamics.
 At each interior equilibrium point, the maximum $\1-$ and  $\2-$utilities must cross.
% For a population with $\p$ types, using induction, it can be shown that there are at most $2\p-1$ intersection points and, in turn,
%  at most $2\p +1$ equilibrium points for the abstract dynamics.
As for minimum number of equilibria, $u^\1_i(x)$ and $u^\2_i(x)$, for $i \in [\p]$, are \tb{polynomials, and hence, and } continuous functions of $x$. 
\tb{Therefore, $\Delta u(x)$, defined as $\max_{i \in [\p]} u^\2_i(x) - \max_{i \in [\p]} u^\1_i(x)$, is  continuous in $x$.
In addition, in view of \Cref{ass:2}, we have $\Delta u(0)>0$  and $\Delta u(1)<0$.
With this, and in view of continuity of  $\Delta u(x)$,
at least one point $x$ on the unit interval exists such that $\Delta u(x) = 0$.
This implies that, including the extreme points, there are at least three equilibria.}
\end{proof} 

Let  $\mathtt{q}$ be the number of equilibrium points of the abstract dynamics. \tb{($\mathtt{q}$ is finite as the utilities are polynomials of degree, at most, $\mathtt{d}$.)}
We then rank these equilibria in ascending order, where $q_i$, for $i \in [\mathtt{q}]$, refers to the $i^{\text{th}}$ equilibrium point,
resulting in
 $q_1 = 0$ and $q_{\mathtt{q}}  = 1$.
     \tb{Note that \Cref{ass:2} ensures that, in a neighborhood of each intersection point, 
\emph{(i)} a unique type attains the maximum $\1$-utility, and a unique type attains 
the maximum $\2$-utility; and \emph{(ii)} the function
$
\max_{i\in[\p]} u^\1_i(x) - \max_{i\in[\p]} u^\2_i(x)
$
is locally equal to the difference of two polynomials, is differentiable at the 
intersection point, and has a simple root there.}
 %Inspired by \cite{FU2024111354}, I removed it on June 12, 2025, because I am going to rewrite it in terms of utilities' parameters, and then it is not any more quite similar to FU.
 \tb{An equilibrium point $q$ is \emph{attracting} if it is either extreme or the relation
   $\dot{\bar{u}}^\2(q) > \dot{\bar{u}}^\1(q)$ holds, where  $\bar{u}^\2(q) = \max_{i\in[\p]} u^\2_i(q)$ and $\bar{u}^\1(q) = \max_{i\in[\p]} u^\1_i(q)$.
   Otherwise, the equilibrium point $q$ is
    \emph{non-attracting}.
}
   Assume that there are ``$\mathtt{a}$'' attracting equilibrium points for the abstract dynamics.
   We denote the $k^\text{th}$ smallest attracting equilibrium point by $q^*_k$, resulting $q^*_1 = 0$ and $q^*_{\mathtt{a}} = 1.$
\begin{lemma}
    Consider the abstract dynamics \eqref{eq:abstract}.
    Under \Cref{ass:2}, between any two consecutive attracting equilibrium points $q^*_{k}$ and $q^*_{k+1}$, for $k \in [\mathtt{a}-1]$, there is one non-attracting equilibrium point $q^*_{k,k+1}$.
\end{lemma}

\Cref{lem:asym-stablility-abstract} reveals the global convergence analysis of the abstract dynamics.
\begin{lemma} \label{lem:asym-stablility-abstract}
    For the  dynamics \eqref{eq:abstract}.
    Under \Cref{ass:2},
    \begin{enumerate}
    \item 
    the attracting equilibrium point $q^*_k$, $k \in  [\mathtt{a}]$, is asymptotically stable with the basin of attraction 
          \begin{equation}
    \mathcal{A}(q^*_k) = 
        \begin{cases}
            [0, q^*_{1,2}), & \text{if } k=1, \\
           (q^*_{k-1,k},q^*_{k,k+1}), & \text{if } 2 \leq k \leq \mathtt{a}-1, \\
            (q^*_{\mathtt{a}-1,\mathtt{a}},1], & \text{if } k= \mathtt{a}. 
        \end{cases}
    \end{equation} 
    \item 
The limit set of the equilibrium point $q^*_{k,k+1}$, for $k \in [\mathtt{a}-1],$ is $\{ q^*_{k},q^*_{k,k+1}, q^*_{k+1}\}$.
\end{enumerate}
\end{lemma}
\begin{proof}
Part 1)
We prove for $k \notin \{1,\mathtt{a}\}.$
The case $k \notin \{1,\mathtt{a}\}$ can be handled similarly.
Consider the function $V(x) = 0.5(x - q^*_k)^2$.
It is straightforward to show that  for $x \in (q^*_{k-1,k}, q^*_k)$ (resp. $x \in (q^*_{k}, q^*_{k,k+1})$), the abstract dynamics read as
     $\dot{x} = 1 -x$ (resp. $\dot{x} =  -x$), refer to the proof of \cite[Lemma A1]{aghaeeyan2023discrete}.
    Consequently, $\dot{V}(x) <0$ for $x \in (q^*_{k-1,k},q^*_{k,k+1})$ where $x \neq q^*_k$ and $\dot{V}(q^*_k) = 0$.
    As a result, the point $q^*_k$ is an asymptotically stable equilibrium point for the abstract dynamics \cite[Theorem 1]{cortes2008discontinuous}.
    As for the basin of attraction, for all $x \in (q^*_{k-1,k},q^*_{k})$ (resp. $x \in (q^*_{k-1,k},q^*_{k})$),
    with regard to $\dot{x} = 1-x$ (resp. $\dot{x} = -x$),
    the evolution of $x(t)$ equals $x(t) = (x_0 -1)\exp(-t) + 1$ (resp. $x(t) = x_0\exp(-t)$), where $x_0$ is the value of the abstract state at $t = 0$.
    Consequently, $x(t)$  falls in an $\epsilon$-neighborhood of $q^*_k$ in  finite time  $t_1 < \ln ((1- x_0)/(1 - q^*_k))$ (resp. $t_1 < \ln (x_0/q^*_k)$).
    In view of $\dot{V}(x) < 0$ for all $x \in (q^*_{k-1,k},q^*_{k,k+1})$ where $x \neq q^*_k$, the solution trajectory remains in the $\epsilon$-neighborhood for all $t \geq t_1.$
    Part 2) The abstract dynamics at $q^*_{k,k+1}$ equals 
    $\dot{x} \in [-x,1-x].$ 
 This implies that the abstract state may either \emph{(i)} remain at $q^*_{k,k+1}$ indefinitely, or \emph{(ii)} depart from $q^*{k,k+1}$ at some point.
In the first case, the point $q^*{k,k+1}$ belongs to its own limit set.
In the second case, the trajectory lies in the basin of attraction of one of the adjacent attracting equilibrium points, $q^*_{k}$ or $q^*_{k+1}$, and therefore converges to one of them.
Thus, the limit set of the point $q^*{k,k+1}$ consists of $q^*_{k}$, $q^*{k,k+1}$, and $q^*_{k+1}$
    This completes the proof.
%       Define  $h(x) =x - q_k$ and $\Sigma = \{x \in [0,1]\vert h(x) = 0\}$.
%      Based on \eqref{eq:abstract}, the  dynamics for $h(x) < 0$ (resp. $h(x) > 0$) is $\dot{x} = 1 - x$ (resp. $\dot{x} =  - x$).
%      Denote the normal vector of $h(\x)$ by $\bm n_h$, which is equal to $\bm 1$.
%      The sign of $\bm n^{\top}_h (\bm \rho - \x) = 1- \alpha$ is positive, and the sign of $\bm n^{\top}_h ( - \x) = - \alpha$ is negative, implying that $h(\x) = 0$ is an attracting  surface for the population dynamics \eqref{eq:type_mixed}, i.e.,
%      once the state reaches $h(\x)$, it cannot leave it \cite{dieci2011sliding}.
%      Up to now, it has been shown that provided $\x_0 \in \prod_{j=1}^{{\p}}(0,{\rho}_j)$, the population dynamics \eqref{eq:type_mixed} converge to $h(\x) = 0$ and remain there afterwards.
%      For the trajectory to remain at $h(\x) = 0$, the value of  $\bm n_h^\top \dot{\x}(t)$ must be equal to zero.
%      The population dynamics at $h(\x)$ read as $\dot{\x} \in [1 - c(\x)]\bm \rho  - \x$
% where $c(\x) \in [0,1]$ satisfies $\bm n^{\top}_h\dot{\x} = 0$. 
% This yields
%    $ c(\x) = \frac{\bm n^{\top}_h (\bm \rho - \x)}{\bm n^{\top}_h(\bm \rho)}.$
% Substituting $\bm n^{\top}_h$ with $\bm 1$ results in 
%     $c(\x) = 1-\alpha$, and, in turn, 
%        $\dot{\x} = \alpha \bm \rho - \x$. 
%        The evolution of the population dynamics reads as $x_p(t) = \big(x_p(t_1) - \alpha \rho_p\big)\exp\big(-(t-t')\big) + \alpha \rho_p$ for $p \in [\p]$, which implies that $\x(t)$ converges to $\alpha \bm \rho$. 
\end{proof} 

% \begin{cor}
%     Under \Cref{ass:2}, the equilibrium point $q_1$ (resp. $q_{\mathtt{q}}$) is attracting and asymptotically stable
%     with basin of attraction $\mathcal{A}(q_1) = (q_1,q_2)$ (resp. $\mathcal{A}(q_1) = (q_{\mathtt{q}-1},q_{\mathtt{q}})$)
%     if $q_2$ (resp. $q_{\mathtt{q}-1}$) is a not-attracting equilibrium point.
%     Otherwise the equilibrium point $q_1$ (resp. $q_{\mathtt{q}}$)  is a non-attracting equilibrium point, and the equilibrium point $q_2$ (resp. $q_{\mathtt{q-1}}$)  is an attracting equilibrium point and asymptotically stable with basin of attraction $\mathcal{A}(q_2) = (q_1,q_3)$ (resp. $\mathcal{A}(q_{\mathtt{q}-1}) = (q_{\mathtt{q}-2},q_{\mathtt{q}})$).
% \end{cor}
\setcounter{revisit}{0}
\begin{revisit}
 The associated abstract dynamics with the finite imitation population dynamics introduced in Example \ref{exampleFinite} read as
 \tb{
  \begin{equation*}
            \dot{x} \in \begin{cases} 
        \{- x\}, & \text{if }   x \in (0, 0.018) \cup (0.509, 0.656),    \\
        \{1 - x\}, & \text{if } x \in (0.018, 0.509) \cup (0.656, 1),\\
         [- x, 1 - x],& \text{otherwise.}  \\
        \end{cases}
\end{equation*}
}
The abstract dynamics admit \tb{$5$} equilibrium points: \tb{$q_1 = 0$, $q_2 = 0.018$, $q_3 = 0.509$, $q_4 = 0.656$, $q_5 = 1$.}
\tb{The points $q_2$ and $q_4$ are unstable, and the others are asymptotically stable.
}
\end{revisit}
\subsection{Continuous-time Population Dynamics}
Here, using the results of the previous section, we obtain the equilibrium points and investigate the asymptotic behavior of the continuous-time population dynamics \eqref{eq:semicontinuous}.
% \subsection{Equilibrium Points and Convergence Analysis}
\begin{lemma} \label{lem_one2one}
   Under Assumptions \ref{ass:2} and \ref{ass:making-A-chosen-strategy-at-1},
   the continuous-time
imitation population dynamics \eqref{eq:semicontinuous} admit the same number of equilibrium points as the associated abstract dynamics \eqref{eq:abstract} and are characterized by
 $q_k \bm \rho$ for $k \in [\mathtt{q}]$.
\end{lemma}
\begin{proof}
There exists no point $\x \in \bm{\mathcal{X}}_{s}$ that satisfies either the first or second case in \eqref{eq:semicontinuous} while also satisfying $\bm{0} \in \bm{\mathcal{V}}(\x)$.
Hence, we should investigate the last case to obtain the equilibrium points of the \eqref{eq:semicontinuous}.
The last case in \eqref{eq:semicontinuous} is applied to the boundary points and the points $\x$ at which
\begin{equation} \label{eq_proof}
\max_{i\in[\p]} u^\1_i(\bm 1^\top \x) =\max_{i\in[\p]} u^\2_i(\bm 1^\top \x).
\end{equation}
The equality \eqref{eq_proof} holds at the intersection points of the maximum $\1-$ and  $\2-$utilities.
On the other hand, the maximum $\1-$ and  $\2-$utilities 
at the equilibrium points of the abstract dynamics also intersect (\Cref{lem:equilibriumOfAbstract}).
This implies that the value of the inner product $\bm 1^\top \x$ for $\x$ satisfying \eqref{eq_proof} will equal exactly one of the abstract equilibrium points $q_k$, $k \in \{2,3,\ldots, \mathtt{q}-1\}$.
% ascending order where $q_2$ is the smallest and $q_{\mathtt{q}-1}$ is the largest value.
% Define $q_1 = 0$ and $q_{\mathtt{q}} = 1$.

If $ \bm 0 \in  \mathrm{Conv}(-\x, \bm \rho -\x)$,
the point $\x$ is an equilibrium point for \eqref{eq:semicontinuous}.
The relation $ \bm 0 \in  \mathrm{Conv}(-\x, \bm \rho -\x)$ 
is satisfied if there exists some $\lambda^* \in [0,1]$ such that $\bm 0 = \lambda^* \bm \rho - \x$.
The inner product of the relation $\bm 0 = \lambda^* \bm \rho - \x$ and an all-one vector $\bm 1$ results in 
$\lambda^* = \bm 1^\top \x$.
\tb{At the points of the intersection, $\lambda^*$ will be equal to $q_k$.}
Therefore, the points $q_k \bm \rho$ are the equilibrium points of the continuous-time dynamics.
\tb{Taking $\lambda^*$ equal to $0$ and $1$, shows that $\bm 0$ and $\bm \rho$ are, respectively, also equilibria.}
Now, we show by contradiction that, excluding the points $\bm 1$ and $\bm \rho$, the remaining boundary points are not equilibria for \eqref{eq:semicontinuous}.
 Assume that $\x^* \in \partial \bm{\mathcal {X}}_{s}$ is an equilibrium where $x^*_i =  0$ and $x^*_j \neq 0$ for some  $i,j \in [\p]$.
 This implies that for some $\lambda \in [0,1]$, we have $\bm 0 = \lambda \bm \rho - \x^*$ and, in turn, $\lambda \rho_l - x^*_l =0$ for $l = 1,2,\ldots, \p$.
 The condition $\lambda \rho_i - x^*_i =0$ requires $\lambda =0$.
 But a zero-valued parameter $\lambda$  results in $\lambda \rho_j-x^*_j \neq 0$.
 This contradicts the assumption that the point $\x^*$ is an equilibrium.
 The boundary points where  $x^*_i =  \rho_i$ and $x^*_j \neq \rho_j$ for some  $i,j \in [\p]$ can be handled similarly.
 This completes the proof.
    % In view of \eqref{eq:semicontinuous}, at every other equilibrium point, the abstract state, $x_e$, should satisfy $\mathtt{s}(x_e) = \{\1,\2\}$.
    % This implies that the abstract state
    % $x_e$ should be  at  one of the interior equilibrium points, i.e., $x_e = q_k$ for some $k$ where $k= 2,3,\ldots,\mathtt{q}-1$.
    % This yields $\sum_{i \in [\p]} x_i = x_e$.
    % In addition, at equilibrium point $\bm x_e$ we should have
    % $\bm 0 \in [\bm 0, \bm \rho] - \x_e.$
    % Note that, the last case in \eqref{eq:semicontinuous} is indeed the convex hull of the two vector fields $\bm \rho - \x$ and $-\x$ where for an arbitrarily $\alpha \in [0,1]$  results in
    % $\alpha \bm \rho = \x_e$. 
    % In view of $\bm x_e^\top \bm 1 = q_k$, we conclude that $\alpha = q_k$.
    % Hence, the point $q_k \bm \rho$ is the equilibrium point of \eqref{eq:semicontinuous}.
    % Since the last case in  \eqref{eq:semicontinuous} is the convex hull of the two vector fields, the semicontinuous dynamics at $x_e = q_k$ cannot admit any other equilibrium points as other than $q_k \bm \rho$ and  
    % this completes the proof.
\end{proof}
 % Inspired by \cite{FU2024111354}, we call an equilibrium point $q_k\bm \rho \notin \{\bm 0, \bm \rho \}$ \emph{attracting} if 
 %  $\exists \delta >0 \forall \varepsilon \in (0,\delta)
 %    \max_{i \in [\p]} u^\1_i(q_k-\varepsilon) >  \max_{i \in [\p]} u^\2_i(q_k-\varepsilon).$       
 %    Similarly, we call an equilibrium point $q_k\bm \rho \notin \{0,1\}$ \emph{not-attracting} if 
 %     $\exists \delta >0 \forall \varepsilon \in (0,\delta)
 %    \max_{i \in [\p]} u^\2_i(q_k-\varepsilon) >  \max_{i \in [\p]} u^\1_i(q_k-\varepsilon).$   

 The next three lemmas describe the asymptotic behavior of the continuous-time dynamics for different initial conditions.
\begin{lemma} 
\label{lem:converges-semicontinuous-no-boundary}
Under Assumptions \ref{ass:2} and \ref{ass:making-A-chosen-strategy-at-1}, starting from initial condition $\x_0 \in \bm{\mathcal{I}}(k)$ where  
        $$\bm{\mathcal{I}}(k) =  \{\bm x  \in \mathrm{int}(\bm{\X}_{s}) \vert \bm 1^\top \x \in (q^*_{k-1,k}, q^*_{k,k+1})\}, \text{ for } k \in [\mathtt{a}],$$
        the trajectories of the continuous-time population dynamics \eqref{eq:semicontinuous}
        converge to 
      $\bm \rho q^*_k$.
\end{lemma}
% In view of \eqref{eq:semicontinuous} and \Cref{ass:1}, the solution to the semicontinuous dynamics with initial condition $\x(0)$ satisfying $\bm 1^\top \x(0) \in \mathcal{A}(q_k)$ can be explicitly calculated. 
% The proof of the above lemma is then based on the definition of the stability.
\begin{proof}
  Given $\x_0 \notin \partial \bm{\X}_{s}$, it can be shown that, for any finite time, the set of solution for the abstract state is the same as that of the population proportion of $\1-$players, $\bm 1^\top \x$ (\Cref{prop_apendix})
  % implying that the population proportion of cooperators $x = \bm 1^\top \x$ and the abstract state are equal.
     Then, in view of \eqref{eq:semicontinuous},
     if the initial abstract state lies in the basin of attraction of the equilibrium point  $q^*_k$, henceforth denoted by $\alpha$,
      the dynamics read as $\dot{\x} = \bm \rho - \x$ (resp. $\dot{\x} =  - \x$) for
      $\x^{\top}_0\bm 1 < \alpha$ (resp. $\x^{\top}_0\bm 1 >\alpha$), and, accordingly, the trajectory of the abstract state equals $x(t) = 1 + (x_0 - 1)e^{-t}$ (resp. $x(t) = x_0 e^{-t}$).
      The abstract state will reach $\alpha$ at  $t_1$, where $t_1 = \ln{\frac{1-x_0}{1-\alpha}}$ (resp.  $t_1 = \ln\frac{x_0}{\alpha}$), \tb{and hence the population state will reach the}
      hyperplane $\Sigma$ defined as $\{\x \in \mathbb{R}^\p \vert  \x^{\top}\bm 1 - \alpha = 0\}$ at $t_1$.
     The normal vector of $\Sigma$ is equal to $\bm 1$. 
     \tb{(\emph{i})} The inner product of the normal vector and the vector field for $\x^\top \bm 1 < \alpha$ is positive
     ($\tb{f_1:}\bm 1^\top (\bm \rho - \x) >0$), and
     \tb{(\emph{ii})} the inner product of the normal vector and the vector field for $\x^\top \bm 1 > \alpha$ is negative
     ($\tb{f_2:} \bm 1^\top ( - \x) <0$).
     \tb{(\emph{iii}) In addition, the vector fields $\bm \rho - \x$, $-\x$, and the equation defining hyperplane $\Sigma$ ($\x^{\top}\bm 1 - \alpha = 0$) are smooth. Hence, in view of \cite[Theorem 2, p. 110]{filippov}, the solution to the population dynamics is unique.
     It remains to construct the unique solution on the hyperplane.
     Given $f_1 >0$ and $f_2 <0$, we construct a solution according to which once the trajectory reaches $\Sigma$, it remains there afterwards }\cite[section 2.1, p. 2026]{dieci2009sliding}.
     In this regard, the value of  $\bm 1^\top \dot{\x}(t)$ must be equal to zero so the trajectory does not leave the hyperplane.
     Given, the dynamics at $\Sigma$ read as $\dot{\x} \in \mathrm{Conv}(\bm \rho - \x, -\x)$, \tb{or, equivalently $\dot{\x} = \lambda \bm \rho - \x$ for some $\lambda \in [0,1]$, we have $\bm 1^\top (\lambda \bm \rho - \x) =0$. Given $\bm 1^\top \x$ on $\Sigma$ equals $\alpha$,}
 it yields $\dot{\x}(t) = \alpha \bm \rho - \x(t)$. 
     \tb{This trajectory is absolutely continuous and satisfies \eqref{eq:semicontinuous} for almost every $t$, hence it is a (unique) solution.  }
       The evolution of the proportion of $\1-$players in each type $i$, $i \in [\p]$, is then equal to  $x_i(t) = \big(x_i(t_1) - \alpha \rho_i\big)\exp\big(-(t-t_1)\big) + \alpha \rho_i$.
       which implies that the term $\vert \x(t) -\alpha \bm \rho \vert$ for $t >t_1$ is strictly decreasing and converges to zero exponentially.
This completes the proof.
\end{proof} 
\begin{lemma} \label{lem:covergence_unstable_no_boundary}
    Under Assumptions \ref{ass:2} and \ref{ass:making-A-chosen-strategy-at-1},
    starting from initial condition $\x_0 \in \bm{\mathcal{R}}(k)$  where
        $$\bm{\mathcal{R}}(k) = \{\bm x  \in \mathrm{int}(\bm{\X}_{s}) \vert \bm 1^\top \x = q^*_{k,k+1}\}, \text{  for $k \in [\mathtt{a}-1],$}$$
        the trajectories of the continuous-time population dynamics \eqref{eq:semicontinuous} converge to $q^*_{k,k+1} \bm \rho$, $q^*_k \bm \rho$, or  $q^*_{k+1} \bm \rho$.
\end{lemma}
\begin{proof}
  Given $\x_0 \notin \partial \bm{\X}_{s}$, for any finite time, the evolution of the abstract state is the same as that of the population proportion of $\1-$players, $\bm 1^\top \x$.
  The initial state belongs to the hyperplane $\Sigma = \{\x \in \mathbb{R}^\p \vert  \x^{\top}\bm 1 - q^*_{k,k+1} = 0\}$.
   The normal vector of $\Sigma$ is equal to $\bm 1$. 
   The inner product of the normal vector and the vector field for $\x^\top \bm 1 < \alpha$ is negative  ($\bm 1^\top ( - \x) =  -\alpha$), and the inner product of the normal vector and the vector field for $\x^\top \bm 1 > \alpha$ is positive ($\bm 1^\top (\bm \rho - \x) =1 - \alpha$).
    This implies that the solution trajectories is not unique  \cite[section 2.1, p. 2026]{ dieci2011sliding}. 
    Three case can happen:
    the trajectories can \emph{(i)} follow the vector field $\bm \rho - \x$, \emph{(ii)}  follow $-\x$, or \emph{(iii)} slide along the surface $\Sigma$.
    In  cases \emph{(i)} and \emph{(ii)}, the abstract state enters $\mathcal{A}(q^*_{k+1})$ and $\mathcal{A}(q^*_{k})$, respectively.
     By \Cref{lem:converges-semicontinuous-no-boundary}, the trajectories therefore converge to $q^*_{k+1}\bm \rho$ in case \emph{(i)} and to $q^*_{k}\bm \rho$ in case \emph{(ii)}.
     In case \emph{(iii)}, when the trajectory of the abstract state remains on $\Sigma$, then, based on a similar argument provided in the proof of  \Cref{lem:converges-semicontinuous-no-boundary},
     the continuous-time dynamics reduce to $q^*_{k,k+1}\bm \rho - \x$, and, consequently, the trajectory approaches $q^*_{k,k+1}\bm \rho$.
     This completes the proof.
\end{proof}
\begin{lemma} \label{lem:convergence_boundary}
     Under Assumptions \ref{ass:2} and \ref{ass:making-A-chosen-strategy-at-1}, the trajectories of the continuous-time population dynamics starting from the initial condition $\x_0 \in \partial \bm{\X}_{s}$  converge to  $q_k \bm \rho$ for some $k \in [\mathtt{q}]$.
\end{lemma}
\begin{proof}
For an initial condition $\x_0 \in \partial \bm{\X}_{s}$, let $\mathcal{P}_0$ \tb{be the set of indices of the zero entries of $\x_0$}, and let $\mathcal{P}_1$ be the set of indices at which $\x_0$ and $\bm \rho$ have equal entries. 
   The relation $\x_0 \in \partial \bm{\X}_{s}$ yields  $\mathcal{P}_{0} \cap \mathcal{P}_{1} \neq \emptyset$  and that the initial state lies on the intersection, denoted $\Sigma$, of the following hyperplanes
    $\Sigma_i = \{\x \in \bm{\X}_{s} \mid x_i=0 \}$ for $i \in \mathcal{P}_0$ and $\Sigma_j = \{\x \in \bm{\X}_{s} \mid x_j= \rho_j \}$ for $j \in \mathcal{P}_1$.
    % In a neighborhood of the intersection of these surfaces, there are $2^\p$ regions characterized based on the value of $x_i$ and $x_j$ for $i \in \mathcal{P}_0$ and $j \in \mathcal{P}_1$.
    % For those regions not belonging to $\bm{\X}_{s}$, following 
    % \cite[2.A.4]{sandholm2010population}, we consider the closet point projection of $\mathbb{R}^\p$ onto the set $\bm{\X}_{s}$ defined by $\Pi_{\bm{\X}_{s}}: \mathbb{R}^\p \to \bm{\X}_{s}$, where $\Pi_{\bm{\X}_{s}}(\y) = \text{argmin}_{\x \in \bm{\X}_{s}} |\y - \x|$.
    % This projection results that the vector fields in these regions belong to the conve
    % % $\{\x \in \mathbb{R}^\p \mid x_i<0, x_j < \rho_j,  i \in \mathcal{P}_0,j\in \mathcal{P}_1 \}$,
    % % $\{\x \in \mathbb{R}^\p \mid x_i<0, x_j > \rho_j,  i \in \mathcal{P}_0,j\in \mathcal{P}_1 \}$,
    % % $\{\x \in \mathbb{R}^\p \mid x_i>0, x_j < \rho_j, ,  i \in \mathcal{P}_0,j\in \mathcal{P}_1\}$, or $\{\x \in \mathbb{R}^\p \mid x_i>0, x_j > \rho_j,  i \in \mathcal{P}_0,j\in \mathcal{P}_1  \}$.
    % Although except for the 
    % The sliding motion on the intersection of these surfaces would require $\bm n^\top_i \bm f_F = 0$ for every $i \in \mathcal{P}_0 \bigcup \mathcal{P}_1$, where $\bm f_F$ belongs to the convex hull of the vector fields in the neighborhood of the intersection. 
    % T
    Part 1) Assume that either of the sets $\mathcal{P}_0$ or $\mathcal{P}_1$, but not both, is nonempty.
    In this case, the dynamics lie in the convex hull of vector fields $-\x$ and $\bm \rho - \x$, i.e., $\dot{\x} = \lambda \bm \rho - \x$ for some $\lambda \in [0,1]$.
    It can be shown that
    if $\mathcal{P}_1 \neq \emptyset$ (resp.  $\mathcal{P}_0 \neq \emptyset$ ) and the trajectories remain on $\partial \bm{\X}_{s}$ for all $t>0$, the dynamics will be equal to $\bm \rho - \x$ (resp.  $  - \x$) and,  accordingly, the trajectory will converge to $\bm \rho$ (resp. $\bm 0$).
    Otherwise, the trajectories at some point will leave $\partial \bm{\X}_{s}$, and 
 according to \Cref{lem:converges-semicontinuous-no-boundary}, will converge to $q_k \bm \rho$ for some $k \in \{2,3,\ldots,\mathtt{q}\}$.
   Part 2) Now, assume that both of these sets are nonempty.
    Then, either
    \emph{(i)} the trajectories remain on $\Sigma$ for all $t >0$, 
    or \emph{(ii)} the trajectories leave  $\Sigma$ at some finite time $t >0$.
    In case \emph{(i)}, the dynamics lie in the convex hull of $-\x$ and $\bm \rho -\x$ resulting $\dot{\x} = \lambda \bm \rho - \x$, for some $\lambda \in [0,1]$.
    The value of $\lambda$ should be equal to $0$ if the trajectories were to remain on hyperplane $\Sigma_i$, whereas the value of $\lambda$ should equal $1$ to have the trajectories remain on hyperplane $\Sigma_j$.
    The contradiction implies that the trajectory cannot remain on $\Sigma$.
    Hence, case  \emph{(ii)} must occur, implying that the trajectories will leave the intersection of hyperplanes characterized by $x_i =0$ for some $i \in \mathcal{P}_0$ and $x_j = \rho_j$ for some $j \in \mathcal{P}_1$.
    Once the trajectories leave either types of hyperplanes, a similar argument to that in Part 1 implies one of the two possible conclusions: either trajectories also leave the other type of hyperplanes and converge to $q_k \bm \rho$ for some $k \in \{2,3,\ldots,\mathtt{q}-1\}$, or they converge to one of the extreme equilibrium points.
 This completes the proof.
\end{proof}

\begin{proposition} \label{prop:birkhoffcenter}
      Under Assumptions \ref{ass:2} and \ref{ass:making-A-chosen-strategy-at-1}, the Birkhoff center of the dynamical system induced by the continuous-time imitation population dynamics \eqref{eq:semicontinuous} is 
$
     \bigcup_{k \in [\mathtt{q}]} \{q_k \bm \rho \}. 
     $
\end{proposition}
\begin{proof}
    Given Lemmas \ref{lem:converges-semicontinuous-no-boundary}-\ref{lem:convergence_boundary}, the recurrent points of the continuous-time dynamics are the same as the equilibria.
    With this and the fact that the set of isolated points is closed,  the Birkhoff center is equal to the set of equilibrium points.
    This completes the proof.
\end{proof}
% \setcounter{revisit}{0}
% \begin{revisit}
% The equilibrium points of the continuous-time dynamics are  $\bm q_1 = \bm 0$, $\bm q_2 = 0.376 \bm \rho$, and $\bm q_3 = \bm \rho$.
% % The equilibrium points  $\bm q_2$ and $\bm q_4$ are unstable and the remaining are asymptotically stable.
% The population proportion of cooperators at the initial condition $\x_0 = (14/33, 5/33, 7/33)^\top$
% belongs to $\mathcal{A}(q_3)$, and, consequently, the continuous-time dynamics approach 
%  $\bm q_3$.
% \end{revisit}
\subsection{The asymptotic behavior of the discrete imitation population dynamics}
Having obtained the Birkhoff center, we link the asymptotic behavior of the discrete imitation population dynamics, as population size approaches infinity, with that of the associated continuous-time dynamics.
\begin{theorem} \label{thm:2}
    Consider the discrete imitation population dynamics for a population of size $\N$ \eqref{eq:discretePopulationDynamics}.
    Under Assumptions \ref{ass:2}-\ref{ass:making-A-chosen-strategy-at-1}, for any open set $\bm{\mathcal{O}}$ containing the Birkhoff center 
    and any 
     sequence $\langle \frac{1}{\N} \rangle$ approaching zero, we have $\lim_{ \frac{1}{\N} \to 0} {\mu}^{\frac{1}{\N}}(\bm{\mathcal{O}})=1$, where ${\mu}^{\frac{1}{\N}}$ is an
    invariant probability measure of the Markov chain that models the discrete imitation population dynamics in a population of size $\N$.
\end{theorem}
\begin{proof}
%\tb{See the arXiv version for a proof \cite{arxiv_imitation}.}
\tb{
 It is straightforward to show that 
  the sequence  $\langle \x^{\N}(k) \rangle$, which  follows the discrete population dynamics \eqref{eq:discretePopulationDynamics},
  is a realization of the Markov chain $\langle \mb{X}^{\frac{1}{\N}}_k \rangle_{k=0}^\infty$ defined in \Cref{prop:markovchain}.
   The Markov chain is homogeneous, and, for each finite $\N$, it is defined over a finite  state space.  Thus, 
 invariant probability measures ${\mu}^{\frac{1}{\N}}$ for Markov chain $\langle\mathbf{X}^{\frac{1}{\N}}_k\rangle_k$ exist.
Based on \Cref{lem_GSAP}, the collection of   Markov chains indexed by the population size
  is a GSAP for  \eqref{eq:semicontinuous}
 whose Birkhoff center is characterized in \Cref{prop:birkhoffcenter}.
 These steps complete the proof 
by defining the sequence $\langle \frac{1}{N}\rangle_{N = \N_0}^{\infty}$ as a vanishing sequence and
\cite[Theorem 1]{aghaeeyan2024discrete}.
}
\end{proof}
What about the fluctuations in the population proportion of $\1-$players?
\begin{cor} \label{cor_fluctuationsDoNotScaleWithN_discretePopulationDynamics_2}
Consider the discrete imitation population dynamics \eqref{eq:discretePopulationDynamics} with a specific initial condition.
     Under the conditions of \Cref{thm:2},
    as the population size approaches infinity,  with probability one,
    the amplitude of the fluctuations in the population proportion of $\1-$players converges to zero.
\end{cor}
% \begin{proof}
%     The proof is similar to that of \cite[Corollary 4]{aghaeeyan2023discrete} and is omitted.
% \end{proof}
\setcounter{revisit}{0}
\begin{revisit}
The maximum and minimum of the population proportion of $\1-$players in \Cref{fig:finitepop1} are indeed getting closer to \tb{$q_3 = 0.509$} for larger population sizes.
\end{revisit}
\begin{remark} \label{rem_tie}
    \tb{
       According to the current model, when both strategies yield the same utilities, the active agent chooses strategy $\1$.
    As shown in \Cref{popA1}, \Cref{lem_GSAP}  (which connects the discrete population dynamics to the mean dynamics) remains valid for some other tie-breaking rules, including favoring strategy $\2$ or choosing between the two strategies uniformly at random. Thus, the results of the paper remain valid under these tie-breaking rules.
    }
\end{remark}

 \section{Concluding Remarks}
We studied the behavior of a heterogeneous population of individuals imitating the highest earners when the population size approaches infinity.
\tb{This setting captures success- or payoff-biased social learning, where agents rely on others' observed success rather than complete payoff information to guide their choices.}
Using the available results in the stochastic approximation theory and through formulating and analyzing the asymptotic behavior of the mean dynamics, it was shown that the amplitudes of the reported perpetual fluctuations in the population proportions of $\1$-players converge to zero with probability one. 
Whether the perpetual fluctuations also diminish in games with more than two available strategies remains a subject for future investigation.
\bibliographystyle{IEEEtran}
\bibliography{IEEEabrv,ref}

\begin{thebibliography}{10}
\providecommand{\url}[1]{#1}
\csname url@rmstyle\endcsname
\providecommand{\newblock}{\relax}
\providecommand{\bibinfo}[2]{#2}
\providecommand\BIBentrySTDinterwordspacing{\spaceskip=0pt\relax}
\providecommand\BIBentryALTinterwordstretchfactor{4}
\providecommand\BIBentryALTinterwordspacing{\spaceskip=\fontdimen2\font plus
\BIBentryALTinterwordstretchfactor\fontdimen3\font minus \fontdimen4\font\relax}
\providecommand\BIBforeignlanguage[2]{{%
\expandafter\ifx\csname l@#1\endcsname\relax
\typeout{** WARNING: IEEEtran.bst: No hyphenation pattern has been}%
\typeout{** loaded for the language `#1'. Using the pattern for}%
\typeout{** the default language instead.}%
\else
\language=\csname l@#1\endcsname
\fi
#2}}

\bibitem{aghaeeyan2024discreteCDC}
A.~Aghaeeyan and P.~Ramazi, ``From discrete to continuous imitation dynamics,'' in \emph{2024 IEEE 63rd Conference on Decision and Control (CDC)}.\hskip 1em plus 0.5em minus 0.4em\relax IEEE, 2024, pp. 966--971.

\bibitem{tanny1988innovators}
S.~M. Tanny and N.~A. Derzko, ``Innovators and imitators in innovation diffusion modelling,'' \emph{Journal of Forecasting}, vol.~7, no.~4, pp. 225--234, 1988.

\bibitem{van2007new}
C.~Van~den Bulte and Y.~V. Joshi, ``New product diffusion with influentials and imitators,'' \emph{Marketing science}, vol.~26, no.~3, pp. 400--421, 2007.

\bibitem{bauch2005imitation}
C.~T. Bauch, ``Imitation dynamics predict vaccinating behaviour,'' \emph{Proceedings of the Royal Society B: Biological Sciences}, vol. 272, no. 1573, pp. 1669--1675, 2005.

\bibitem{10453658}
L.~Arditti, G.~Como, F.~Fagnani, and M.~Vanelli, ``Robust coordination of linear threshold dynamics on directed weighted networks,'' \emph{IEEE Transactions on Automatic Control}, vol.~69, no.~10, pp. 6515--6529, 2024.

\bibitem{10172285}
G.~Chen and Y.~Yu, ``Convergence analysis and strategy control of evolutionary games with imitation rule on toroidal grid,'' \emph{IEEE Transactions on Automatic Control}, vol.~68, no.~12, pp. 8185--8192, 2023.

\bibitem{ramazi2020convergence}
P.~Ramazi and M.~Cao, ``Convergence of linear threshold decision-making dynamics in finite heterogeneous populations,'' \emph{Automatica}, vol. 119, p. 109063, 2020.

\bibitem{ramazi2017asynchronous}
------, ``Asynchronous decision-making dynamics under best-response update rule in finite heterogeneous populations,'' \emph{IEEE Transactions on Automatic Control}, vol.~63, no.~3, pp. 742--751, 2017.

\bibitem{grabisch2020anti}
M.~Grabisch and F.~Li, ``Anti-conformism in the threshold model of collective behavior,'' \emph{Dynamic Games and Applications}, vol.~10, no.~2, pp. 444--477, 2020.

\bibitem{roohi}
P.~Ramazi and M.~H. Roohi, ``Characterizing oscillations in heterogeneous populations of coordinators and anticoordinators,'' \emph{Automatica}, vol. 154, p. 111068, 2023.

\bibitem{FU2024111354}
Y.~Fu and P.~Ramazi, ``Evolutionary matrix-game dynamics under imitation in heterogeneous populations,'' \emph{Automatica}, vol. 159, p. 111354, 2024.

\bibitem{sandholm2010population}
W.~H. Sandholm, \emph{Population games and evolutionary dynamics}.\hskip 1em plus 0.5em minus 0.4em\relax MIT press, 2010.

\bibitem{golman2010basins}
R.~Golman and S.~E. Page, ``Basins of attraction and equilibrium selection under different learning rules,'' \emph{Journal of evolutionary economics}, vol.~20, no.~1, pp. 49--72, 2010.

\bibitem{bestresponsePotential}
B.~Swenson, R.~Murray, and S.~Kar, ``On best-response dynamics in potential games,'' \emph{SIAM Journal on Control and Optimization}, vol.~56, no.~4, pp. 2734--2767, 2018.

\bibitem{berger2007two}
U.~Berger, ``Two more classes of games with the continuous-time fictitious play property,'' \emph{Games and Economic Behavior}, vol.~60, no.~2, pp. 247--261, 2007.

\bibitem{theodorakopoulos2012selfish}
G.~Theodorakopoulos, J.-Y. Le~Boudec, and J.~S. Baras, ``Selfish response to epidemic propagation,'' \emph{IEEE Transactions on Automatic Control}, vol.~58, no.~2, pp. 363--376, 2012.

\bibitem{cianfanelli2025stability}
L.~Cianfanelli and G.~Como, ``On the stability of the logit dynamics in population games,'' \emph{IEEE Transactions on Automatic Control}, 2025.

\bibitem{comoImitation}
G.~Como, F.~Fagnani, and L.~Zino, ``Imitation dynamics in population games on community networks,'' \emph{IEEE Transactions on Control of Network Systems}, vol.~8, no.~1, pp. 65--76, 2021.

\bibitem{replicator}
D.~Madeo and C.~Mocenni, ``Game interactions and dynamics on networked populations,'' \emph{IEEE Transactions on Automatic Control}, vol.~60, no.~7, pp. 1801--1810, 2015.

\bibitem{replicator2}
M.~A. Mabrok, ``Passivity analysis of replicator dynamics and its variations,'' \emph{IEEE Transactions on Automatic Control}, vol.~66, no.~8, pp. 3879--3884, 2021.

\bibitem{nogales2020replicator}
J.~M.~S. Nogales and S.~Zazo, ``Replicator based on imitation for finite and arbitrary networked communities,'' \emph{Applied Mathematics and Computation}, vol. 378, p. 125166, 2020.

\bibitem{10842053}
S.~Wang, M.~Cao, and X.~Chen, ``Optimally combined incentive for cooperation among interacting agents in population games,'' \emph{IEEE Transactions on Automatic Control}, vol.~70, no.~7, pp. 4562--4577, 2025.

\bibitem{benaim2005stochastic}
M.~Bena{\"\i}m, J.~Hofbauer, and S.~Sorin, ``Stochastic approximations and differential inclusions,'' \emph{SIAM Journal on Control and Optimization}, vol.~44, no.~1, pp. 328--348, 2005.

\bibitem{benaim1998recursive}
M.~Bena{\"\i}m, ``Recursive algorithms, urn processes and chaining number of chain recurrent sets,'' \emph{Ergodic Theory and Dynamical Systems}, vol.~18, no.~1, pp. 53--87, 1998.

\bibitem{benaim2003deterministic}
M.~Bena{\"\i}m and J.~W. Weibull, ``Deterministic approximation of stochastic evolution in games,'' \emph{Econometrica}, vol.~71, no.~3, pp. 873--903, 2003.

\bibitem{benaim2009mean}
M.~Bena{\"\i}m and J.~Weibull, ``Mean-field approximation of stochastic population processes in games,'' \emph{hal-00435515}, 2009.

\bibitem{roth2013stochastic}
G.~Roth and W.~H. Sandholm, ``Stochastic approximations with constant step size and differential inclusions,'' \emph{SIAM Journal on Control and Optimization}, vol.~51, no.~1, pp. 525--555, 2013.

\bibitem{aghaeeyan2023discrete}
A.~Aghaeeyan and P.~Ramazi, ``From discrete to continuous binary best-response dynamics: Discrete fluctuations almost surely vanish with population size,'' \emph{arXiv preprint arXiv:2311.01995v2}, 2024.

\bibitem{ramazi2022lower}
P.~Ramazi, J.~Riehl, and M.~Cao, ``The lower convergence tendency of imitators compared to best responders,'' \emph{Automatica}, vol. 139, p. 110185, 2022.

\bibitem{nowakMay1992}
M.~A. Nowak and R.~M. May, ``Evolutionary games and spatial chaos,'' \emph{Nature}, vol. 359, no. 6398, pp. 826--829, 1992.

\bibitem{henrich2001evolution}
J.~Henrich and F.~J. Gil-White, ``The evolution of prestige: Freely conferred deference as a mechanism for enhancing the benefits of cultural transmission,'' \emph{Evolution and human behavior}, vol.~22, no.~3, pp. 165--196, 2001.

\bibitem{BarrettEtAl2017}
B.~J. Barrett, R.~L. McElreath, and S.~E. Perry, ``Pay-off-biased social learning underlies the diffusion of novel extractive foraging traditions in a wild primate,'' \emph{Proceedings of the Royal Society B: Biological Sciences}, vol. 284, no. 1856, p. 20170358, 2017.

\bibitem{cortes2008discontinuous}
J.~Cortes, ``Discontinuous dynamical systems,'' \emph{IEEE Control systems magazine}, vol.~28, no.~3, pp. 36--73, 2008.

\bibitem{MAYHEW20111045}
C.~G. Mayhew and A.~R. Teel, ``On the topological structure of attraction basins for differential inclusions,'' \emph{Systems \& Control Letters}, vol.~60, no.~12, pp. 1045--1050, 2011.

\bibitem{filippov}
A.~Filippov, \emph{Differential Equations with Discontinuous Righthand Sides}.\hskip 1em plus 0.5em minus 0.4em\relax Kluwer Academic Publishers, 1988.

\bibitem{dieci2009sliding}
L.~Dieci and L.~Lopez, ``Sliding motion in filippov differential systems: theoretical results and a computational approach,'' \emph{SIAM Journal on Numerical Analysis}, vol.~47, no.~3, pp. 2023--2051, 2009.

\bibitem{dieci2011sliding}
------, ``Sliding motion on discontinuity surfaces of high co-dimension. a construction for selecting a filippov vector field,'' \emph{Numerische Mathematik}, vol. 117, pp. 779--811, 2011.

\bibitem{aghaeeyan2024discrete}
A.~Aghaeeyan and P.~Ramazi, ``From discrete to continuous best-response dynamics: Discrete fluctuations do not scale with population size,'' in \emph{2024 American Control Conference (ACC)}.\hskip 1em plus 0.5em minus 0.4em\relax IEEE, 2024, pp. 851--856.

\end{thebibliography}
\appendix
\begin{appendixLemma} \label{lem_USC}
The continuous-time population dynamics \eqref{eq:semicontinuous} 
are good upper semicontinuous.
\end{appendixLemma}
\begin{proof}
For each $\x \in \bm{\X}_{s}$, $\bm{\V}(\x)$ is nonempty and bounded.
As for convexity of $\bm{\V}(\x)$, if  $\bm{\V}(\x)$ is a singleton, it is obviously convex; otherwise  $\bm{\V}(\x)$ will read as $\lambda \bm \rho -\x$ for $\lambda \in [0,1]$.
In this case, to prove the convexity, we need to show that if
$\y_1, \y_2 \in \bm{\V}(\x)$ then $\alpha \y_1 + (1-\alpha) \y_2 \in  \bm{\V}(\x)$, for $\alpha \in [0,1]$.
By definition,
$\y_1\in \bm{\V}(\x)$ 
(resp. $\y_2\in \bm{\V}(\x)$)
means $\y_1 = \lambda_1 \bm \rho - \x$
(resp. $\y_2 = \lambda_2 \bm \rho - \x$),
 for some $\lambda_1 \in [0,1]$ (resp.  $\lambda_2 \in [0,1]$).
Substituting $\lambda_1 \bm \rho - \x$ (resp. $\lambda_2 \bm \rho - \x$) for $\y_1$ (resp. $\y_2$) in $\alpha \y_1 + (1-\alpha) \y_2$
results in $(\lambda_1 \alpha + \lambda_2 (1-\alpha)) \bm \rho - \x$, and in view of 
 $(\lambda_1 \alpha + \lambda_2 (1-\alpha)) \in [\min(\lambda_1, \lambda_2), \max(\lambda_1, \lambda_2)]$, we conclude that 
 $\lambda_1 \alpha + \lambda_2 (1-\alpha) \in [0,1]$ and, in turn,
 the term
 $\alpha \y_1 + (1-\alpha) \y_2$ also
 belongs to $\bm{\V}(\x)$. 
 Thus, $\bm{\V}(\x)$ is convex for each $\x \in \bm{\mathcal{X}}_s.$
Now we show that the graph of $\bm {\V}(\x)$ is closed.
    Let the set 
    ${ \mathcal{Q}}_a$ (resp. ${\mathcal{Q}}_n$) consist of the attracting  (resp. non-attracting) equilibrium points of the abstract dynamics \eqref{eq:abstract}.
    Consider the set $\bm {\mathcal{G}}^{+}_k$ for $k=1,2,\ldots, \mathtt{q}-1$
    \begin{align*}
    \bm {\mathcal{G}}^{+}_k:=\{(\x, \y) \mid  &\x \in \bm{\X}_{s}, \bm 1^\top \x \in [q_{k}, q_{k+1}], \y =  - \x 
    \\
    &\text{ if } q_{k}\in { \mathcal{Q}}_a  \text{ else } \y = \bm \rho - \x \},
    \end{align*}
    the set $\bm {\mathcal{G}}^{-}_k$ for $k=2,\ldots, \mathtt{q}$ 
    \begin{align*}
       \bm {\mathcal{G}}^{-}_k:=\{(\x, \y) \mid &  \x \in \bm{\X}_{s},\bm 1^\top \x \in [q_{k-1}, q_{k}], \y =\bm \rho  - \x\\
       &\text{ if } q_{k}\in { \mathcal{Q}}_a  \text{ else } \y =  - \x \},
    \end{align*}
    \begin{align*}
     \bm {\mathcal{G}}_1:=\big\{(\x,\y) \mid \x &\in \partial\bm{\X}_{s}, \bm 1^\top \x \in \{q_1,q_2,\ldots,q_{\mathtt{q}}\}, \\ &\y =\lambda \bm \rho -\x, \text{ $\lambda \in [0,1]$} \big\}, \qquad \text{and}
     \end{align*}
    $$\bm {\mathcal{G}}:=\{(\x,\y) \mid \x \in \partial\bm{\X}_{s}, \y =\lambda \bm \rho -\x, \text{ $\lambda \in [0,1]$}\}.$$
    The sets $\bm {\mathcal{G}}^{-}_k$ and $\bm {\mathcal{G}}^{+}_k$
    are the graphs of continuous functions defined over closed subsets of  $\bm{\X}_{s}$, and, consequently, they are closed.
    The closed-ness of  the sets $\bm {\mathcal{G}}_1$ and $\bm {\mathcal{G}}$ is immediate.
    The union of these sets is the graph of $\bm {\V}(\x)$ which is then closed.
    This completes the proof.
\end{proof}
\tb{A Caratheodory solution for a differential inclusion is any absolutely continuous function that satisfies the differential inclusion almost everywhere \cite{cortes2008discontinuous}.}
\begin{appendixLemma}\label{prop_apendix}
\tb{Let $\x (\cdot)$ be any Caratheodory solution for the continuous-time imitation population dynamics with initial condition $\x_0\in\mathrm{int}(\bm{\mathcal{X}}_s)$, and denote $\mathbf{1}^\top \x(t)$ by $z(t)$.
Then, for every finite horizon $T$,
%\begin{enumerate}[label=(\roman*),leftmargin=2.2em,itemsep=0.25em]
%\item $\x(t)\in \bm{\mathcal{X}}_s$ for all $t\in[0,T]$ (hence all indicator functions in~(7)--(9) equal~$1$ on $[0,T]$); and
 $z(\cdot)$ is a solution for the abstract differential inclusion (12) on $[0, T]$ with initial condition $x(0)=\mathbf{1}^\top \x_0$.
%\end{enumerate}
}
\end{appendixLemma}
\begin{proof}
\tb{
Because $\x_0\in\mathrm{int}(\bm{\mathcal{X}}_s)$, we have $0<x_{i,0}<\rho_i$ for all $i\in [\p]$, where $x_{i,0}$ denotes the $i^\text{th}$ element of $\x_0.$
%By definition of Caratheodory solutions, $\x(\cdot)$ is absolutely continuous and satisfies $\dot{\x}(t)\in V(\mathbf{x}(t))$ for almost all $t$, where the set-valued drift $V(\cdot)$ is given in~(10).
%We consider any time $t$ at which $\dot{\mathbf{x}}(t)$ exists and belongs to $V(\mathbf{x}(t))$.
The three cases in (11) can be summarized as
%$\V(\x)\subseteq \mathrm{conv}\{\boldsymbol{\rho}-\x,\,-\mathbf{x}\}$ for every $\mathbf{x}\in\mathcal{X}^s$.
%Therefore, for almost all $t$ there exists a scalar $\lambda(t)\in[0,1]$ such that
\begin{equation}\label{eq:lambda_form}
\dot{\x}(t)=\lambda(t)\bm \rho-\x(t).
\end{equation}
Here $\lambda(t)$ is scalar-valued and bounded on the interval $[0,1]$.
Fix any finite $T<\infty$.
Solving \eqref{eq:lambda_form} gives, for each $i \in [\p]$ and all $t\in[0,T]$,
$
x_i(t)=e^{-t}x_{i,0}+\int_{0}^{t}e^{-(t-\tau)}\lambda(\tau)\rho_i\,d\tau.
$
Since $x_{i,0}>0$, $\rho_i>0$, and $\lambda(\tau)\ge 0$, for all $t\in[0,T]$, we have  $x_i(t)>0$ .
Moreover, $x_{i,0}<\rho_i$ and $\lambda(\tau)\le 1$ results in
$
x_i(t)\le e^{-t}\rho_i+\int_{0}^{t}e^{-(t-\tau)}\rho_i\,d\tau
= e^{-t}\rho_i+\rho_i(1-e^{-t})=\rho_i
$
for all $t\in[0,T]$, and strict inequality holds for all finite $t$ because $x_{i,0}<\rho_i$ and $\lambda(\cdot)$ is bounded above by 1.
Thus, $0<x_i(t)<\rho_i$ for all $t\in[0,T]$ and all $i$, i.e., $\x(t)\in\mathrm{int}(\bm{\mathcal{X}}_s)$ on $[0,T]$.
This implies that the conditions in (11), for $t \in [0,T]$, will only be dependent on $\bm 1^\top \x$ and, accordingly, the conditions
     $\x \notin \partial \bm{\mathcal{X}}_{s} \text{ and } u^\1(\x) \substack{> \\ <}   u^\2(\x)$ will reduce to $u^\1(\bm 1^\top \x) \substack{> \\ <}   u^\2(\bm 1^\top  \x)
     $ or, equivalently, to
$\displaystyle\max_{i\in[\p]} u^\1_i(\bm 1^\top \x) 
     \substack{> \\ <}
     \displaystyle\max_{i\in[\p]} u^\2_i(\bm 1^\top \x)
$.
Now, define $z(t)= \bm 1^\top \x(t)$.
At $\x$, where 
$ \displaystyle\max_{i\in[\p]} u^\1_i(\bm 1^\top \x) 
     >
     \displaystyle\max_{i\in[\p]} u^\2_i(\bm 1^\top \x),$ we have
     $\dot{\x} = \bm \rho - \x$, and hence, $\dot{z}(t) = 1 - z$. 
     The condition
     $ \displaystyle\max_{i\in[\p]} u^\1_i(\bm 1^\top \x) 
     <
     \displaystyle\max_{i\in[\p]} u^\2_i(\bm 1^\top \x),$
implies 
     $\dot{\x} =  - \x$, and hence, $\dot{z}(t) = - z$, and finally,  $ \displaystyle\max_{i\in[\p]} u^\1_i(\bm 1^\top \x) 
     =
     \displaystyle\max_{i\in[\p]} u^\2_i(\bm 1^\top \x),$ results in $\dot{\x} \in \mathrm{Conv}(\bm \rho - \x)$ and, in turn,
     $\dot{z}(t) \in [-1,1] $.
     Thus, in all three cases, the set of trajectories of $z(\cdot)$ is the same as that of the abstract state.
     }
\end{proof}
\begin{propAppendix}
\label{popA1}
\Cref{lem_GSAP} remains true under the following tie-breaking rules:
(\emph{i}) The active agent 
chooses her strategy uniformly at random,
(\emph{ii}) The active agent chooses strategy $\2$.
\end{propAppendix}
\begin{proof}
\tb{
When a tie happens, both strategies yield the same utilities, i.e., $u^\1(\x^\N) = u^\2(\x^\N).$
We prove for the tie-breaking rule (\emph{i}), the proof for the other rule is straightforward.
Under the first rule, we redefine the function $s(\x)$ as follows 
\begin{align} \label{eq:u_uniformly_random}
{s}(\x^{\N})  =
    \begin{cases}
   1, & \text{if }  u^\1(\x^\N) >  u^\2(\x^\N),    \\
   2, & \text{if }  u^\1(\x^\N) <   u^\2(\x^\N), \\
   \frac{3}{2}, & \text{otherwise.}
    \end{cases}
\end{align}
Conditions 1, 2, and 4 in \Cref{defGSAP} can be verified by following the same steps as in \Cref{lem_GSAP}.
As for the third condition, at $\x^\N$ where 
 $u^\1(\x^\N) = u^\2(\x^\N)$, $\bm \nu (\x^\N) = \bm \rho \big(2 - s(\x^{\N})\big)-\x^{\N}  $ reduces to $\frac{1}{2}\bm \rho - \x^\N$. 
 At $\x^\N$ where $u^\1(\x^\N) > u^\2(\x^\N)$, $\bm \nu (\x^\N)$ is  equal to  $\bm \rho - \x^\N$, and finally at
  $\x^\N$ where $u^\1(\x^\N) < u^\2(\x^\N)$, $\bm \nu (\x^\N)$ is  equal to  $- \x^\N$.
  In all these three cases, $\bm \nu (\x^\N)$  belongs to $\mathrm{Conv} (\bm \rho - \x, -\x)$.
  Thus,
 the third condition is  satisfied by taking $\y = \x$ in \eqref{eq_3cnd_GSAP}.
 This completes the proof.
 }
\end{proof}
\end{document}